\documentclass[10pt]{amsart}

\usepackage[english]{babel}
\usepackage[left=3cm,right=3cm,top=2.5cm,bottom=3cm]{geometry}

\usepackage{amsmath,amssymb,amsthm,bbm,color,graphics,version, mathtools}
\usepackage{mathrsfs}
\usepackage{amsfonts}
\usepackage{graphicx}
\usepackage{setspace}
\usepackage[T1]{fontenc}
\usepackage[english]{babel}
\usepackage[utf8]{inputenc}
\usepackage{mathtools}
\usepackage{dsfont}
\usepackage[shortlabels]{enumitem}
\usepackage{bbm}
\usepackage{url}
\usepackage[dvipsnames]{xcolor}
\usepackage[colorinlistoftodos]{todonotes}
\usepackage{subcaption}

\newtheorem{theorem}{Theorem}
\newtheorem{Theorem}[theorem]{Theorem}
\newtheorem{Assumption}{Assumption}
\newtheorem{Corollary}[theorem]{Corollary}

\newtheorem{Example}[theorem]{Example}

\newtheorem{Proposition}[theorem]{Proposition}
\newtheorem{Remark}[theorem]{Remark}

\title{Pricing Perpetual American put options with asset-dependent discounting}

\author{Jonas Al-Hadad}
\address{Faculty of Pure and Applied Mathematics\\
Wrocław University of Science and Technology\\
ul. Wyb. Wyspiańskiego 27, 50-370 Wrocław \\
Poland}
\email{jonas.al-hadad@pwr.edu.pl}

\author{Zbigniew Palmowski}
\address{Faculty of Pure and Applied Mathematics\\
Wrocław University of Science and Technology\\
ul. Wyb. Wyspiańskiego 27, 50-370 Wrocław \\ Poland}
\email{zbigniew.palmowski@pwr.edu.pl}

\thanks{Jonas Al-Hadad and Zbigniew Palmowski have been partially supported by the National Science Centre under the grant 2016/23/B/HS4/00566.}

\date{\today}

\begin{document}

\begin{abstract}
The main objective of this paper is to present an algorithm
of pricing perpetual American put options with asset-dependent discounting.
The value function of such an instrument can be described as
\begin{equation*}
V^{\omega}_{\text{\rm A}^{\text{\rm Put}}}(s) = \sup_{\tau\in\mathcal{T}}
\mathbb{E}_{s}[e^{-\int_0^\tau \omega(S_w) dw} (K-S_\tau)^{+}],
\end{equation*}
where $\mathcal{T}$ is a family of stopping times, $\omega$ is a discount function and $\mathbb{E}$ is an expectation taken with respect to a martingale measure.
Moreover, we assume that the asset price process $S_t$ is a geometric L\'evy process with negative exponential jumps, i.e. $S_t = s e^{\zeta t + \sigma B_t - \sum_{i=1}^{N_t} Y_i}$.
The asset-dependent discounting is reflected in the $\omega$ function, so this approach is a generalisation of the classic case when $\omega$ is constant. It turns out that under certain conditions on the $\omega$ function, the value function $V^{\omega}_{\text{\rm A}^{\text{\rm Put}}}(s)$ is convex and can be represented in a closed form; see \cite{hadadPalma}.
We provide an option pricing algorithm in this scenario and we present exact calculations for the particular choices of $\omega$ such that $V^{\omega}_{\text{\rm A}^{\text{\rm Put}}}(s)$ takes a simplified form.

\vspace{3mm}


\end{abstract}

\maketitle

\pagestyle{myheadings} \markboth{\sc J.\ Al-Hadad --- Z.\ Palmowski
} {\sc Pricing algorithm of Perpetual American put option with asset-dependent discounting}

\vspace{1.8cm}

\tableofcontents

\newpage

\section{Introduction}
In this paper we consider a perpetual American put option with asset-dependent discounting.
We consider a standard stochastic background for this problem, i.e. we define a complete filtered risk-neutral probability space
$(\Omega, \mathcal{F}, \{\mathcal{F}_t\}_{t\geq 0}, \mathbb{P})$, on which we define the asset price process $S_t$.
Then $\mathcal{F}_t$ is a natural filtration of $S_t$ satisfying the usual conditions and $\mathbb{P}$ is a risk-neutral measure under which the discounted (with respect to the risk-free interest rate $r>0$) asset price process $e^{-rt}S_t$ is a local martingale.
A family of $\mathcal{F}_t$-stopping times is denoted by $\mathcal{T}$ while $\mathbb{E}_s$ denotes the expectation with respect to $\mathbb{P}$ when $S_0=s=e^x$.
The value function of the perpetual American put option with asset-dependent discounting can be represented by
\begin{equation}\label{mainProblem}
V^{\omega}_{\text{\rm A}^{\text{\rm Put}}}(s) := \sup_{\tau\in\mathcal{T}} \mathbb{E}_{s}\left[e^{-\int_0^\tau \omega(S_w) dw} (K - S_\tau)^{+}\right].
\end{equation}
The asset-dependent discounting is reflected in the $\omega$ function what is our key concept considered in this article.
We underline here that the discount function $\omega$ for various economical reasons can be different from the risk-free interest rate $r>0$; see \cite{hadadPalma} for further explanations.
The way we choose discounting is to model strong dependence
of discount factor with the asset price.
The goal is to understand various economical phenomena that might appear in this extreme case. Our approach differs from typical studies considered in the literature, where the interest rate is independent from the asset price or there is a weak dependence between these two factors.
%\footnote{To zdanie wprowadzilem i zastapilo ono zdanie: In particular, we believe that typical weak correlation between stochastic returns and asset prices might be insufficient to get the right price of American option in some particular cases.}
Therefore, this research is noteworthy not only in the context of option pricing, but also in other areas where optimisation problems appear.

Moreover, we assume that the asset price process $S_t$ is a geometric L\'evy process with negative exponential jumps, i.e.
\begin{equation}\label{S_t}
S_t := s e^{X_t}
\end{equation}
with
\begin{equation}\label{X_t}
X_t := \zeta t + \sigma B_t - \sum_{i=1}^{N_t} Y_i,
\end{equation}
where $\zeta$ and $\sigma>0$ are constant, $N_t$ is the Poisson process with intensity $\lambda\geq 0$ independent of Brownian motion $B_t$ and
$\{Y_i\}_{i\in \mathbb{N}}$ are i.i.d. random variables independent of $B_t$ and $N_t$ having exponential distribution with mean $1/\varphi>0$.
Under the martingale measure $\mathbb{P}$ the drift parameter is of the form $$\zeta := r - \frac{\sigma^2}{2} + \frac{\lambda}{\varphi+1}.$$ Note that when $\lambda=0$ then we end up with the classical Black-Scholes model.
However, empirical studies show that stock prices have heavier left tail than normal distribution.
Therefore, nowadays many books and articles concern, as we do in this work, pricing of derivative securities in market models based on
L\'evy processes; see \cite{cont} for more details.

%In general, pricing American options requires solving an optimal stopping problem.
%It stems from the fact that American options can be exercised at any time, so in order to find their fair price, we need to define the optimal exercise time.

The main objective of this paper is to present an algorithm
of pricing perpetual American put options with asset-dependent discounting with the value function defined in \eqref{mainProblem} and the asset price process $S_t$ given in \eqref{S_t}.
%Such choice of the process $S_t$ stems from the fact that it occurs quite common in financial applications and adequately reflects the real %behaviour of stock prices; for more details see \cite{cont}.
Furthermore, we take into account some specific scenarios (e.g. when $\sigma = 0$ or $\lambda=0$) and for these cases we are able to derive analytical forms of the value function, while for more complex examples we show how to handle them numerically.

Detailed theoretical results of the analysed problem was already developed in \cite{hadadPalma}, where the authors presented the approach of deriving a closed form of value function \eqref{mainProblem} for even a more general setting than it is considered here.
Therefore, in this paper we focus more on numerical side of this problem and analyse in detail few particular cases where
more explicit results can be derived.

Still, before we present the option pricing method in our set-up,
we recall the most important theoretical issues on which our article is based on.
A key step in deriving a closed form of \eqref{mainProblem} is identifying the form of the optimal stopping rule $\tau^{*}$ for which the supremum in \eqref{mainProblem} is attained.
It turns out that under certain conditions on the discount function $\omega$, which are presented in the next section, the value function is convex.
By combining this fact with the classical optimal stopping theory presented e.g. in \cite{peskir}, it allows us to conclude that the optimal stopping region is an interval $[l^{*}, u^{*}]$ and hence
$$\tau^{*}=\inf\{t\geq 0: S_t\in [l^*,u^*]\}$$ for some optimal thresholds $l^{*}\leq u^{*}$.
Observe that for the nonnegative discount function $\omega$ we have $l^{*}=0$ (since waiting is not beneficial).
Therefore, in this case a single continuation region appears. In general, for the negative $\omega$ we can observe a double continuation region; for more details see \cite{tumilewicz}.

The optimal boundary levels $l^{*}$ and $u^{*}$ can be found by application of standard methods of maximising the function
\begin{equation*}
v^{\omega}_{\text{\rm A}^{\text{\rm Put}}}(s, l, u):=\mathbb{E}_{s}\left[e^{-\int_0^{\tau_{l,u}} \omega(S_w) dw} (K-S_{\tau_{l,u}})^{+}\right]
\end{equation*}
over $l$ and $u>l$.
To find $v^{\omega}_{\text{\rm A}^{\text{\rm Put}}}(s, l, u)$ we use exit identities for spectrally negative L\'evy processes containing so-called omega scale functions introduced in \cite{LiBo}.

As shown in \cite[Theorem 9]{hadadPalma}, another way of finding the optimal thresholds $l^*<u^*$
is to apply the classical smooth and continuous fit conditions.

Typically, a price of the option is a solution to a certain Hamiltonian-Jacobi-Bellman (HJB) system and the optimal thresholds are identified using the smooth fit conditions.
We want to underline that our approach is different, although still finding the omega scale functions is done via solving certain
ordinary differential equations.

The paper is organised as follows.
In Section \ref{sec:preliminaries} we introduce basic theory and notation.
%assumptions that we use throughout the paper and we recall basic information about the optimal stopping time and scale functions that play a significant role in our deliberations.
Section \ref{sec:mainResults} provides main theoretical results of this paper.
%, i.e theorems showing how to present the value function with the scale functions and how to obtain their forms by solving the appropriate ordinary differential equation.
In Section \ref{sec:optionpricinganalytical} we present some specific examples where the option price can be expressed in the explicit way.
%in an analytical way, while in
Section \ref{sec:optionpricingnumerical} focuses on the purely numerical analysis. We also show there that these two approaches are consistent.
The last section includes our conclusions.

\section{Preliminaries}\label{sec:preliminaries}

\subsection{Assumptions}
It the beginning, we note that the analysed American put option will not be realised when its payoff is equal to $0$.
Hence, we can transform the form of the value function given in \eqref{mainProblem} into the following one
\begin{equation}\label{mainProblem2}
V^{\omega}_{\text{\rm A}^{\text{\rm Put}}}(s) := \sup_{\tau\in\mathcal{T}} \mathbb{E}_{s}\left[e^{-\int_0^\tau \omega(S_w) dw} (K - S_\tau)\right].
\end{equation}
We work under the same assumptions as those formulated in \cite[Section 2.2]{hadadPalma}. However, this time we consider slightly more specific assumptions
%, therefore the only assumptions we impose
on the $\omega$ function, namely that
% are as follows
\begin{Assumption}\label{Assumptions}
A discount function $\omega$ is concave, nondecreasing and bounded from below.
\end{Assumption}

%The crucial property of the value function $V^{\omega}_{\text{\rm A}^{\text{\rm Put}}}(s)$ that allows us to define the optimal stopping time, thereby the optimal stopping region  is its convexity.
From \cite[Remark 3]{hadadPalma} we can then conclude that under Assumption \ref{Assumptions} the value function $V^{\omega}_{\text{\rm A}^{\text{\rm Put}}}(s)$ is convex.

\subsection{Optimal stopping time}
From  \cite[Section 2.4]{hadadPalma} it follows that then the optimal exercise time
is the first entrance of the process $S_t$ into some interval, that is, it has the following form
\begin{equation*}
\tau_{l, u} = \inf\{t\geq 0: S_t\in [l,u]\}.
\end{equation*}
Hence, we can represent value function \eqref{mainProblem2} as
\begin{equation*}
V^{\omega}_{\text{\rm A}^{\text{\rm Put}}}(s) = v^{\omega}_{\text{\rm A}^{\text{\rm Put}}}(s, l^{*}, u^{*}),
\end{equation*}
where
\begin{equation*}
v^{\omega}_{\text{\rm A}^{\text{\rm Put}}}(s, l^{*}, u^{*}):= \sup_{0\leq l\leq u\leq K}v^{\omega}_{\text{\rm A}^{\text{\rm Put}}}(s, l, u)
\end{equation*}
and
\begin{equation}\label{form}
v^{\omega}_{\text{\rm A}^{\text{\rm Put}}}(s, l, u):=\mathbb{E}_{s}\left[e^{-\int_0^{\tau_{l,u}} \omega(S_w) dw} (K-S_{\tau_{l,u}})\right].
\end{equation}
Moreover, we denote the optimal stopping time by
\begin{equation*}
\tau^{*} := \tau_{l^{*}, u^{*}},
\end{equation*}
where $l^{*}$ and  $u^{*}$ realise the supremum above.
As shown in \cite[Theorem 9]{hadadPalma}, another way of identifying the critical points $l^*$ and $u^*$ can be done via application of the smooth fit property. In that case, $l^*$ and $u^*$ satisfy
\begin{equation}
(V^{\omega}_{\text{\rm A}^{\text{\rm Put}}})'(u^{*}) = -1 \quad\text{and}\quad (V^{\omega}_{\text{\rm A}^{\text{\rm Put}}})'(l^{*}) = -1.
\end{equation}

%Finally, we can simply note that if the discount function $\omega$ is nonnegative, then it is never optimal to wait
%to exercise the option for small asset prices, that is, always $l^{*} = 0$.
%We focus on this case later in the article.

\subsection{Scale functions}\label{scaleFunctions}
By applying the fluctuation theory of
L\'evy processes, we can find a closed form of \eqref{form} and hence of \eqref{mainProblem2} in terms of the so-called omega scale functions.

%Let us briefly recall some important definitions of the classical scale functions and their generalisations.

To introduce them formally, firstly let us define the Laplace exponent of $X_t$ via
\begin{equation*}
\psi(\theta) := \frac{1}{t}\log \mathbb{E}[e^{\theta X_t}\mid X_0=0],
\end{equation*}
which is well-defined for $\theta\geq 0$ since our $X_t$ is a spectrally negative L\'evy process.
In the case of $X_t$ given in \eqref{X_t} the Laplace exponent takes the form
\begin{equation}\label{laplaceExponent}
\psi(\theta) = \zeta \theta + \frac{\sigma^2}{2}\theta^2 - \frac{\lambda\theta}{\varphi+\theta}.
\end{equation}
By
$\Phi(q)$ we denote the right inverse of $\psi(\theta)$, i.e.
\begin{equation*}
\Phi(q):=\sup\{\theta\geq 0: \psi(\theta) = q\},
\end{equation*}
where $q\geq 0$.

The first scale function $W^{(q)}(x)$ is defined as a continuous and increasing function such that $W^{(q)}(x) = 0$ for all $x<0$, while for $x\geq 0$ it is defined via the following Laplace transform
\begin{equation}\label{laplaceTransform}
\int_0^{\infty} e^{-\theta x} W^{(q)}(x) dx = \frac{1}{\psi(\theta)-q}
\end{equation}
for $\theta>\Phi(q)$.
We define also the related scale function $Z^{(q)}(x)$ by
\begin{equation*}
Z^{(q)}(x) := 1 + q \int_0^x W^{(q)}(y) dy,
\end{equation*}
where $x\in\mathbb{R}$.
From \cite{kuznetsov} we know that for $X_t$ given in \eqref{X_t} we have
\begin{equation}\label{W^q(x)}
W^{(q)}(x) = \frac{e^{\gamma_1 x}}{\psi'(\gamma_1)} + \frac{e^{\gamma_2 x}}{\psi'(\gamma_2)} + \frac{e^{\Phi(q) x}}{\psi'(\Phi(q))},
\end{equation}
where $\{\gamma_1, \gamma_2, \Phi(q)\}$ is the set the real solutions to $\psi(\theta) = q$.
In turn, $Z^{(q)}(x)$ is as follows
\begin{equation}\label{Z^q(x)}
Z^{(q)}(x) = 1 + q \bigg(\frac{e^{\gamma_1 x} - 1}{\gamma_1\psi'(\gamma_1)} + \frac{e^{\gamma_2 x} - 1}{\gamma_2\psi'(\gamma_2)} + \frac{e^{\Phi(q) x} - 1}{\Phi(q)\psi'(\Phi(q))}\bigg).
\end{equation}
If we take $\sigma = 0$ or $\lambda = 0$ in \eqref{laplaceExponent} then
$W^{(q)}(x)$ and $Z^{(q)}(x)$ take simplified forms
\begin{equation*}
W^{(q)}(x) = \frac{e^{\gamma_1 x}}{\psi'(\gamma_1)} + \frac{e^{\gamma_2 x}}{\psi'(\gamma_2)}
\end{equation*}
and
\begin{equation*}
Z^{(q)}(x) = 1 + q \bigg(\frac{e^{\gamma_1 x} - 1}{\gamma_1\psi'(\gamma_1)} + \frac{e^{\gamma_2 x} - 1}{\gamma_2\psi'(\gamma_2)}\bigg)
\end{equation*}
for $\gamma_1$ and $\gamma_2$ being again the real solutions to $\psi(\theta) = q$.

The generalisation of $W^{(q)}(x)$ and $Z^{(q)}(x)$ are the $\xi$-scale functions $\{\mathcal{W}^{(\xi)}(x), x\in\mathbb{R}\}$, $\{\mathcal{Z}^{(\xi)}(x), x\in\mathbb{R}\}$, where $\xi$ is an arbitrary measurable function. They are defined as the unique solutions to the following equations
\begin{align}
\mathcal{W}^{(\xi)}(x) &= W(x) + \int_0^{x} W(x - y)\xi(y)\mathcal{W}^{(\xi)}(y)dy,\label{scaleFunctionW} \\
\mathcal{Z}^{(\xi)}(x) &= 1 + \int_0^{x} W(x - y)\xi(y)\mathcal{Z}^{(\xi)}(y)dy\label{scaleFunctionZ},
\end{align}
where $W(x)=W^{(0)}(x)$ is a classical zero scale function.

To simplify notation, we introduce also the following $S_t$ counterparts of the scale functions \eqref{scaleFunctionW} and \eqref{scaleFunctionZ}
\begin{align}
\mathscr{W}^{(\xi)}(s) &:= \mathcal{W}^{(\xi\circ {\rm exp})}(\log s),\label{scaleFunctionW2} \\
\mathscr{Z}^{(\xi)}(s) &:= \mathcal{Z}^{(\xi\circ {\rm exp})}(\log s), \label{scaleFunctionZ2}
\end{align}
where $\xi\circ {\rm exp} (x):= \xi(e^x)$.

For $\alpha$ for which the Laplace exponent is well-defined we can define a new probability measure $\mathbb{P}^{(\alpha)}$ via
\begin{equation*}
\left.\frac{d\mathbb{P}^{(\alpha)}_s}{d\mathbb{P}_s}\right\vert_{\mathcal{F}_t} = e^{\alpha (X_{t}-\log s) - \psi(\alpha)t}.
\end{equation*}
By \cite{exponentialMartingale} and \cite[Cor. 3.10]{kyprianou}, under $\mathbb{P}^{(\alpha)}$, the process $X_t$ is again spectrally negative L\'evy process with the new Laplace exponent
%\begin{equation*}
%\psi^{(\alpha)}(\theta):=\psi(\theta+\alpha)-\psi(\alpha).
%\end{equation*}
%In our case it takes the form
\begin{equation*}
\psi^{(\alpha)}(\theta):= \psi(\theta+\alpha)-\psi(\alpha)=\zeta^{(\alpha)}\theta +\frac{{\sigma^{(\alpha)}}^2}{2}{\theta^{(\alpha)}}^2 - \frac{\lambda^{(\alpha)}\theta^{(\alpha)}}{\varphi^{(\alpha)}+\theta^{(\alpha)}},
\end{equation*}
where
\begin{equation}\label{newparameters}
\zeta^{(\alpha)} := \zeta + \sigma^2\alpha,\quad
\sigma^{(\alpha)} := \sigma,\quad
\lambda^{(\alpha)}:=\frac{\lambda\varphi}{\varphi+\alpha}\quad\text{and}\quad \varphi^{(\alpha)}:=\varphi+\alpha.
\end{equation}

For the new probability measure $\mathbb{P}^{(\alpha)}$ we can define the $\xi$-scale functions
which are denoted by the adding subscript $\alpha$ to the regular counterparts, i.e. $\mathscr{W}^{(\xi)}_{\alpha}(s)$, $\mathscr{Z}^{(\xi)}_{\alpha}(s)$.

Lastly, we define the following auxiliary functions
\begin{align}
&%\eta(x) := \omega(e^x) = \omega(s), \quad \eta_u(x):=\eta(x+\log u), \quad
\omega_u(s):=\omega(su) \quad \text{and}\quad%\\ &\iota(x):=\eta_u(x) - \psi(\alpha)\quad  \text{and} \quad
\omega_u^\alpha(s):=\omega_u(s)- \psi(\alpha).\label{omegaalpha}
\end{align}

\subsection{Theoretical representation of the price}\label{sec:mainResults}
The starting point for our entire analysis are the following results.
The first one is a corollary from \cite[Theorem 15]{hadadPalma}.
%Following on \cite[Theorem 12]{hadadPalma} we can state a theorem about closed forms of the value function $V^{\omega}_{\text{\rm A}^{\text{\rm Put}}}(s)$ that arise in our setting.
%
\begin{Theorem}\label{mainTheorem}
Let Assumption \ref{Assumptions} holds and assume that $\omega$ is nonnegative. % and $S_t$ is given in \eqref{S_t}.
Then the optimal stopping region is of the form $(0, u^*]$ and  we have
\leavevmode
\begin{enumerate}
\item
For $\sigma=0$ and $\lambda>0$
\begin{align}\label{sigma=0}
V^{\omega}_{\text{\rm A}^{\text{\rm Put}}}(s) := \sup_{u>0}v^{\omega}_{\text{\rm A}^{\text{\rm Put}}}(s, 0, u) = \sup_{u>0}\bigg\{\left(K -
\frac{u\varphi}{\varphi+1}\right)\left(\mathscr{Z}^{(\omega_u)}\left(\frac{s}{u}\right)- c_{\mathscr{Z}^{(\omega)}/\mathscr{W}^{(\omega)}} \mathscr{W}^{(\omega_u)}\left(\frac{s}{u}\right)\right)\bigg\}.
\end{align}

\item
For $\lambda=0$ and $\sigma>0$
\begin{equation}\label{lambda=0}
\begin{aligned}
V^{\omega}_{\text{\rm A}^{\text{\rm Put}}}(s) &:= \sup_{u>0}v^{\omega}_{\text{\rm A}^{\text{\rm Put}}}(s, 0, u) \\ &= \sup_{u>0}\bigg\{(K-u)
\left(\lim_{\alpha\rightarrow\infty} \left(\frac{s}{u}\right)^{\alpha}\left(\mathscr{Z}^{(\omega_u^\alpha)}_{\alpha}\left(\frac{s}{u}\right) -
c_{\mathscr{Z}^{(\omega^\alpha)}_{\alpha}/\mathscr{W}^{(\omega^\alpha)}_{\alpha}}
\mathscr{W}^{(\omega_u^\alpha)}_{\alpha}\left(\frac{s}{u}\right)\right)\right)\bigg\}.
\end{aligned}
\end{equation}

\item
For $\sigma>0$ and $\lambda>0$
\begin{equation}\label{sigma>0}
\begin{aligned}
V^{\omega}_{\text{\rm A}^{\text{\rm Put}}}(s) &:= \sup_{u>0}v^{\omega}_{\text{\rm A}^{\text{\rm Put}}}(s, 0, u) = \sup_{u>0}\bigg\{\left(K -
\frac{u \varphi}{\varphi+1}\right)\left(\mathscr{Z}^{(\omega_u)}\left(\frac{s}{u}\right)- c_{\mathscr{Z}^{(\omega)}/\mathscr{W}^{(\omega)}} \mathscr{W}^{(\omega_u)}\left(\frac{s}{u}\right)\right)
\\ &+ (K-u)
\left(\lim_{\alpha\rightarrow\infty} \left(\frac{s}{u}\right)^{\alpha}\left(\mathscr{Z}^{(\omega_u^\alpha)}_{\alpha}\left(\frac{s}{u}\right) -
c_{\mathscr{Z}^{(\omega^\alpha)}_{\alpha}/\mathscr{W}^{(\omega^\alpha)}_{\alpha}}
\mathscr{W}^{(\omega_u^\alpha)}_{\alpha}\left(\frac{s}{u}\right)\right)\right)\bigg\},
\end{aligned}
\end{equation}

where
\begin{equation*}
c_{\mathscr{Z}^{(\omega)}/\mathscr{W}^{(\omega)}} = \lim_{z\rightarrow\infty}\frac{\mathscr{Z}^{(\omega)}(z)}{\mathscr{W}^{(\omega)}(z)}
\quad
\text{and}
\quad
c_{\mathscr{Z}^{(\omega^\alpha)}_{\alpha}/\mathscr{W}^{(\omega^\alpha)}_{\alpha}} = \lim_{z\rightarrow\infty}\frac{\mathscr{Z}^{(\omega^\alpha)}_{\alpha}(z)}{\mathscr{W}^{(\omega^\alpha)}_{\alpha}(z)}.
\end{equation*}

\end{enumerate}

The optimal boundary $u^{*}$ in \eqref{lambda=0} and \eqref{sigma>0} can be determined by the smooth fit condition
\begin{equation*}
(V^{\omega}_{\text{\rm A}^{\text{\rm Put}}})'(u^*) = -1,
\end{equation*}
while the optimal boundary $u^{*}$ in \eqref{sigma=0} can be determined by the continuous fit condition
\begin{equation*}
V^{\omega}_{\text{\rm A}^{\text{\rm Put}}}(u^*) = K - u^*.
\end{equation*}

\end{Theorem}
%For the proof of the above theorem we refer to \cite[Proof of Theorem 15]{hadadPalma}.
\begin{Remark}\label{valueFunctionInsOrInx}\rm
Let us note that using \eqref{scaleFunctionW2} and \eqref{scaleFunctionZ2} we can interpret the value functions occurring in Theorem \ref{mainTheorem} both as the functions of $s$ variable and $x$ variable, where $x=\log s$.
\end{Remark}

\begin{Remark}\label{TwoComponentsOfTheSum}\rm
Formula \eqref{lambda=0} corresponds to the continuous transition of $S_t$ to an interval $(0, u]$, while \eqref{sigma=0} describes the situation when $S_t$ jumps from $(u, \infty)$ into $(0, u]$. A combination of these two components makes up formula \eqref{sigma>0}.
\end{Remark}

To find the option price $V^{\omega}_{\text{\rm A}^{\text{\rm Put}}}(s)$ we have to identify
$$\mathscr{W}^{(\omega_u^\alpha)}_\alpha(s)= \mathcal{W}^{(\omega_u^\alpha\circ {\rm exp})}_\alpha(\log s)\quad\text{and}\quad
\mathscr{Z}^{(\omega_u^\alpha)}_\alpha(s) = \mathcal{Z}^{(\omega_u^\alpha\circ {\rm exp})}_\alpha(\log s),$$
where
$\omega_u^\alpha(s)$ is given in \eqref{omegaalpha} and the case of $\alpha=0$ corresponds to
$\mathscr{W}^{(\omega_u)}(s)$ and $\mathscr{Z}^{(\omega_u)}(s)$.

Observe that %in both cases $\lambda=0$, $\sigma>0$ and $\sigma>0$, $\lambda>0$
we need to find the above $\xi$-scale functions for $\xi =\omega_u^\alpha\circ {\rm exp}$ under measure $\mathbb{P}^{(\alpha)}$, i.e. we have to identify $\mathcal{W}^{(\xi)}_{\alpha}(x)$ and $\mathcal{Z}^{(\xi)}_{\alpha}(x)$. This is equivalent to
taking our asset price process $S_t$ of the form of \eqref{S_t} but with the new parameters given in \eqref{newparameters}.

The second key result for our numerical analysis follows straightforward from \cite[Theorem 16]{hadadPalma} and allows to identify the above omega scale functions using ordinary differential equations.

%From the definition of the first scale function given in \eqref{laplaceTransform} for $q = 0$ and based on formulas \eqref{scaleFunctionW} and \eqref{scaleFunctionZ} we can easily conclude that the $\xi$-scale functions $\mathcal{W}^{(\xi)}(x)$ and $\mathcal{Z}^{(\xi)}(x)$ are the solution to ordinary differential equations of second order for $\sigma=0$ or
%$\lambda=0$ and of third order for $\sigma>0$ (and $\lambda>0$), respectively.
%This fact is stated in Theorem \ref{mainTheoremODE} and forms the basis of the option pricing algorithm presented in the next sections.

\begin{Theorem}\label{mainTheoremODE}
Assume that $\xi$ is continuously differentiable. % and $S_t$ is given in \eqref{S_t}.
Then
\leavevmode
\begin{enumerate}
\item
If $\sigma=0$ and $\lambda>0$ or $\lambda=0$ and $\sigma>0$ then $\mathcal{W}^{(\xi)}(x)$ solves
\begin{equation}\label{secondOrderODE}
{\mathcal{W}^{(\xi)}}''(x) = ((\Upsilon_1+\Upsilon_2)\xi(x)+\gamma_2){\mathcal{W}^{(\xi)}}'(x) + ((\Upsilon_1+\Upsilon_2)\xi'(x) - \Upsilon_1\gamma_2 \xi(x))\mathcal{W}^{(\xi)}(x)
\end{equation}
with
\begin{equation*}
\begin{cases}
{\mathcal{W}^{(\xi)}}(0) = \Upsilon_1 + \Upsilon_2, \\
{\mathcal{W}^{(\xi)}}'(0) = \Upsilon_2\gamma_2 + (\Upsilon_1+\Upsilon_2)^2\xi(0).
\end{cases}
\end{equation*}
Moreover, $\mathcal{Z}^{(\xi)}(x)$ solves the same equation \eqref{secondOrderODE} with
\begin{equation*}
\begin{cases}
{\mathcal{Z}^{(\xi)}}(0) = 1, \\
{\mathcal{Z}^{(\xi)}}'(0) = (\Upsilon_1+\Upsilon_2)\xi(0).
\end{cases}
\end{equation*}

\item
If $\sigma>0$ and $\lambda>0$ then $\mathcal{W}^{(\xi)}(x)$ solves
\begin{equation}\label{thirdOrderODE}
\begin{split}
{\mathcal{W}^{(\xi)}}'''(x) &= \left(\gamma_2+\gamma_3\right){\mathcal{W}^{(\xi)}}''(x) \\&+ \left(\Upsilon_2(\gamma_2-\gamma_3)\xi(x) - \gamma_2\gamma_3 - \Upsilon_1\gamma_3\xi(x)\right){\mathcal{W}^{(\xi)}}'(x) \\&+ \left(\Upsilon_2(\gamma_2-\gamma_3)\xi'(x) + \Upsilon_1\gamma_2\gamma_3 \xi(x) - \Upsilon_1\gamma_3 \xi'(x)\right){\mathcal{W}^{(\xi)}}(x)
\end{split}
\end{equation}
with
\begin{equation*}
\begin{cases}
{\mathcal{W}^{(\xi)}}(0) = 0, \\
{\mathcal{W}^{(\xi)}}'(0) = \Upsilon_2\gamma_2 + \Upsilon_3\gamma_3, \\
{\mathcal{W}^{(\xi)}}''(0) = \Upsilon_2 {\gamma_2}^2 + \Upsilon_3 {\gamma_3}^2.
\end{cases}
\end{equation*}
Moreover, $\mathcal{Z}^{(\xi)}(x)$ solves the same equation \eqref{thirdOrderODE} with
\begin{equation*}
\begin{cases}
{\mathcal{Z}^{(\xi)}}(0) = 1, \\
{\mathcal{Z}^{(\xi)}}'(0) = 0, \\
{\mathcal{Z}^{(\xi)}}''(0) = \xi(0)(\Upsilon_2(\gamma_2-\gamma_3) - \Upsilon_1 \gamma_3).
\end{cases}
\end{equation*}

\end{enumerate}

\end{Theorem}

%The proof of Theorem \ref{mainTheoremODE} is stated in \cite[Proof of Theorem 13]{hadadPalma}

%\begin{remark}\rm

%\end{remark}

\section{Option pricing -- analytical approach}\label{sec:optionpricinganalytical}

In this section we present some examples of discount functions for which we are able to determine the analytical form of the value function $V^{\omega}_{\text{\rm A}^{\text{\rm Put}}}(s)$.

\subsection{Constant discount function}\label{constantOmegaNumerical}
The case when $\omega$ function is constant, i.e. $\omega(s)=q$ is the standard example which appears in the literature quite extensively.
However, this case is quite special, as it turns out that the second term of the sum in \eqref{sigma>0} simplifies and we do not need to deal with the measure $\mathbb{P}^{\alpha}$ (and thus to calculate the limit for $\alpha\rightarrow\infty$) to find $V^{\omega}_{\text{\rm A}^{\text{\rm Put}}}(s)$. This fact is stated in the below theorem.

\begin{Theorem}\label{TheoremAboutIdentityForConstantOmega}
Assume that $\omega(s)=q$.
Then
\begin{equation}\label{identityForConstantOmega}
\begin{aligned}
&\lim_{\alpha\rightarrow\infty} \left(\frac{s}{u}\right)^{\alpha}\left(\mathscr{Z}^{(q-\psi(\alpha))}_{\alpha}\left(\frac{s}{u}\right) -
c_{\mathscr{Z}^{(q-\psi(\alpha))}_{\alpha}/\mathscr{W}^{(q-\psi(\alpha))}_{\alpha}}
\mathscr{W}^{(q-\psi(\alpha))}_{\alpha}\left(\frac{s}{u}\right)\right) \\ &\qquad= \frac{\sigma^2}{2}\left(\mathscr{W}^{(q)\prime}\left(\frac{s}{u}\right) - \Phi(q)\mathscr{W}^{(q)}\left(\frac{s}{u}\right)\right).
\end{aligned}
\end{equation}

\end{Theorem}

\begin{proof}
Note that
\begin{equation}\label{limitInTheProof}
\lim_{\alpha\rightarrow\infty} \left(\frac{s}{u}\right)^{\alpha}\left(\mathscr{Z}^{(q-\psi(\alpha))}_{\alpha}\left(\frac{s}{u}\right) -
c_{\mathscr{Z}^{(q-\psi(\alpha))}_{\alpha}/\mathscr{W}^{(q-\psi(\alpha))}_{\alpha}}
\mathscr{W}^{(q-\psi(\alpha))}_{\alpha}\left(\frac{s}{u}\right)\right)
\end{equation}
corresponds to the continuous transition of the process $S_t$ to the interval $(0, u]$ or, in other words, continuous exit from half-line $(u,\infty)$.
We define
\begin{equation*}
\tau_b^+ = \inf\{t\geq 0: X_t \geq b\}, \quad \tau_0^- = \inf\{t\geq 0: X_t \leq 0\}.
\end{equation*}
It turns out that formula \eqref{limitInTheProof} is equivalent to
\begin{equation*}
\mathbb{E}_{\frac{s}{u}}\left[e^{-q\tau_0^-}; \tau_0^-<\tau_b^+, X_{\tau_0^-}=0\right],
\end{equation*}
for details see \cite[Proof of Theorem 7]{hadadPalma}.

Then using \cite[(13)]{Loeffen} for $x = \log\left(\frac{s}{u}\right)$, $a = 0$ and $v^{(q)}(x) = {W}^{(q)\prime}(x)$ together with the fact that ${W}^{(q)\prime}(0) = \frac{2}{\sigma^2}$ (see \cite[8.5 (ii), p. 235]{kyprianou}) we obtain
\begin{equation*}
\mathbb{E}_{\frac{s}{u}}\left[e^{-q\tau_0^-}; \tau_0^-<\tau_b^+, X_{\tau_0^-}=0\right]=
\frac{\sigma^2}{2} \left({W}^{(q)\prime}(x-\log u) - \frac{{W}^{(q)\prime}(b)}{{W}^{(q)}(b)}{W}^{(q)}(x-\log u)\right).
\end{equation*}
Lastly, taking limit $b\uparrow \infty$ and applying L'Hospital's Rule complete the proof.
\end{proof}

Ultimately, value function \eqref{sigma>0} for the constant discount function $\omega(s) = q$ can be written as
\begin{equation*}
\begin{aligned}
&V^{\omega}_{\text{\rm A}^{\text{\rm Put}}}(s) = \sup_{u>0}\bigg\{\left(K -
\frac{u \varphi}{\varphi+1}\right)\left({\mathscr{Z}}^{(q)}\left(\frac{s}{u}\right)- c_{{\mathscr{Z}}^{(q)}/{\mathscr{W}}^{(q)}} {\mathscr{W}}^{(q)}\left(\frac{s}{u}\right)\right)
\\ &+ (K-u)
\frac{\sigma^2}{2}\left({\mathscr{W}}^{(q)\prime}\left(\frac{s}{u}\right)-\Phi(q){{\mathscr{W}}}^{(q)}\left(\frac{s}{u}\right)\right)\bigg\}.
\end{aligned}
\end{equation*}

\begin{Remark}\label{BSConstandOmega}\rm
For the case of $\lambda = 0$, using \eqref{identityForConstantOmega}, one can show that value function \eqref{lambda=0} simplifies to the well known formula for the value function in the Black-Scholes model, i.e.
\begin{equation*}
V^{\omega}_{\text{\rm A}^{\text{\rm Put}}}(s) = \sup_{u>0}\bigg((K-u) \left(\frac{s}{u}\right)^{-\frac{2r}{\sigma^2}}\bigg),
\end{equation*}
where we substituted $q = r$ and $\zeta = r - \frac{\sigma^2}{2}$. Therefore, we are not forced to apply the smooth fit condition in order to find the optimal value of $u$. We can do this analytically by finding the maximum of $V^{\omega}_{\text{\rm A}^{\text{\rm Put}}}(s)$ with respect to $u$.

\end{Remark}

Figure \ref{ThreeDifferentq} presents value function \eqref{sigma>0} for three different values of $q$, i.e. $q\in\{0.3, 0.6, 0.9\}$.

\begin{figure}[ht]
\centering
\includegraphics[width=13cm]{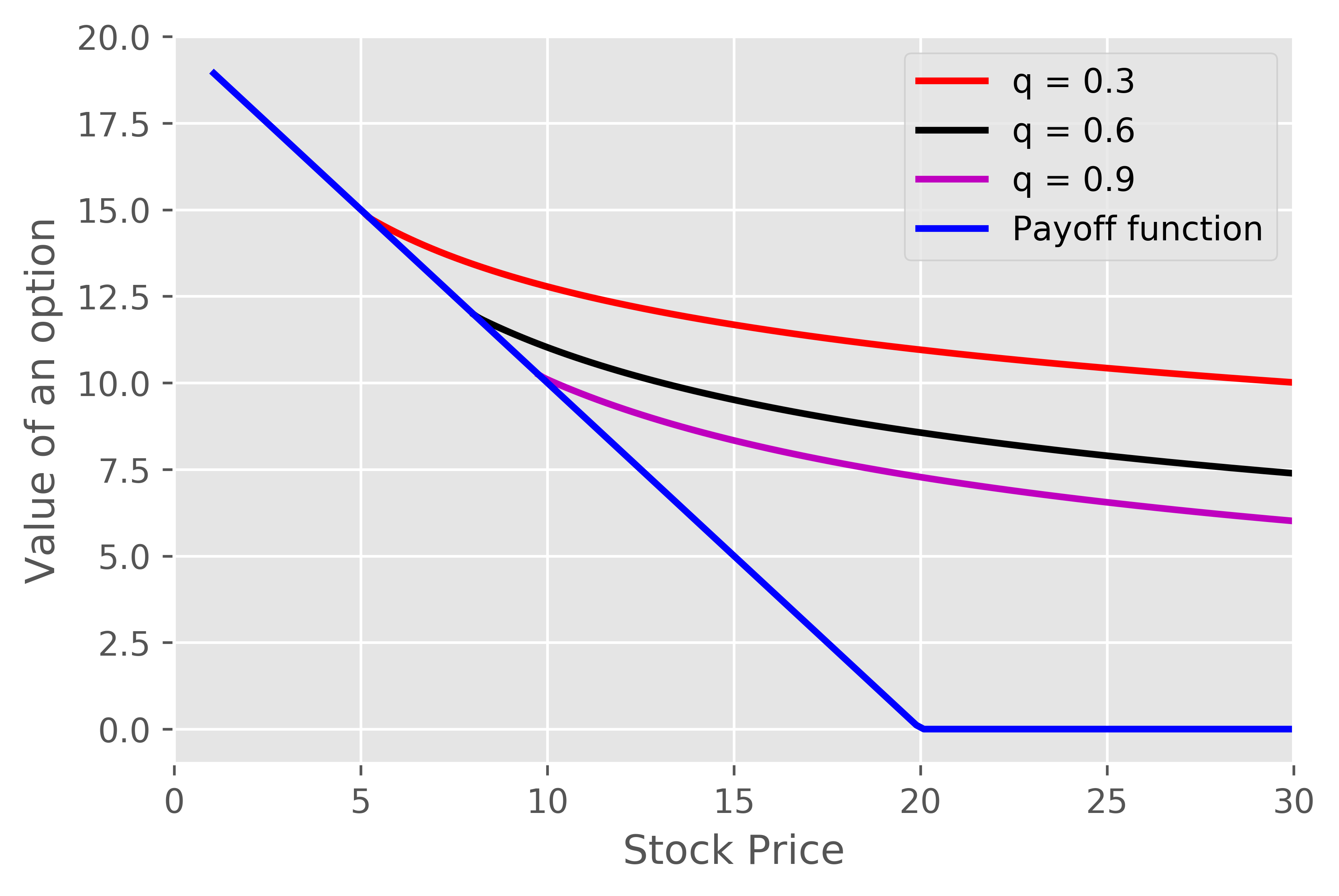}
\caption{The value and payoff functions for the given set of parameters: $K=20$, $r=0.05$, $\sigma=0.2$, $\lambda = 6$, $\varphi = 2$ and $q\in\{0.3, 0.6, 0.9\}$.}
\label{ThreeDifferentq}
\end{figure}

Based on Figure \ref{ThreeDifferentq} we can simply note that a higher value of the discount function $\omega$ results in a smaller value of $V^{\omega}_{\text{\rm A}^{\text{\rm Put}}}(s)$ which is in line with the financial intuition.

In turn, Figure \ref{ThreeModelsForGivenq} shows a comparison of \eqref{sigma=0}, \eqref{lambda=0} and \eqref{sigma>0} for the same value of $q = 0.5$.
\begin{figure}[ht]
\centering
\includegraphics[width=13cm]{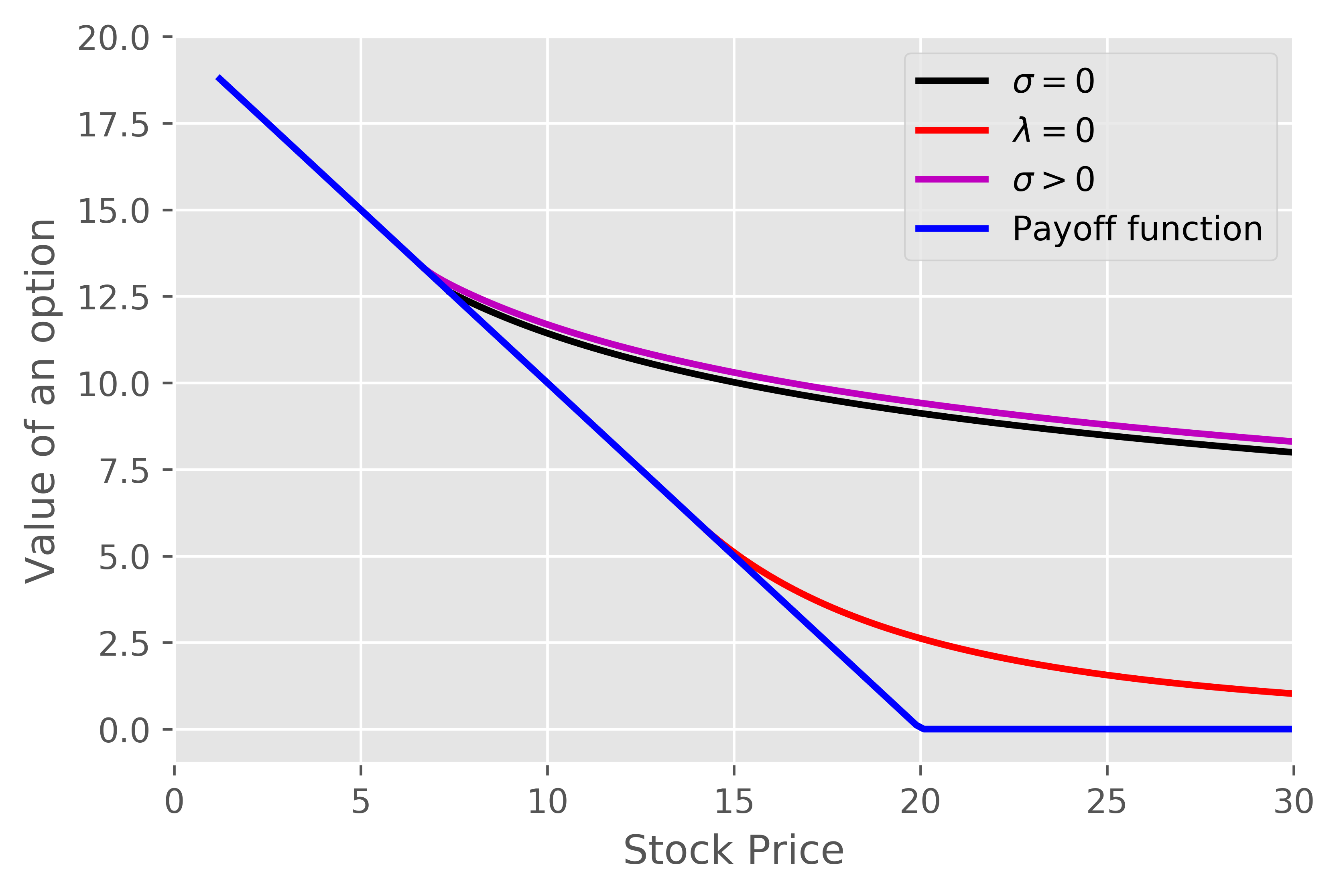}
\caption{The value and payoff functions for the given set of parameters: $K=20$, $r=0.05$, $\sigma=0.4$, $\lambda = 6$, $\varphi = 2$ and $q = 0.5$.}
\label{ThreeModelsForGivenq}
\end{figure}

The resulting relation between these functions is again consistent with the economical expectations.

\subsection{Linear discount function}\label{linearOmegaNumerical}

In this subsection we consider a linear discount function of the form $\omega(s) = Cs$ for some positive constant $C$.

\subsubsection{{\bf $\sigma = 0$}}\label{sigma0Subsection}
% 70, 558/1487, Handbook of exact solutions for ODEs

% https://www.wolframalpha.com/input/?i=f%27%27%28x%29+%3D+%28A+e%5Ex+%2B+B%29f%27%28x%29+%2B+D+e%5Ex+f%28x%29

Let us consider the case of $\sigma=0$. Then the asset price process $S_t$ jumps into the interval $(0,u]$, which means
from Theorem \ref{mainTheorem} that
\begin{equation}\label{valueFunctionSigma0Ins}
V^{\omega}_{\text{\rm A}^{\text{\rm Put}}}(s)=\sup_{u>0} \left(K -
\frac{u\varphi}{\varphi+1}\right)\left(\mathscr{Z}^{(\omega_u)}\left(\frac{s}{u}\right)- c_{\mathscr{Z}^{(\omega)}/\mathscr{W}^{(\omega)}} \mathscr{W}^{(\omega_u)}\left(\frac{s}{u}\right)\right),
\end{equation}
where $\omega_u\left(\frac{s}{u}\right) = \omega(s) = Cs$.
Equivalently, \eqref{valueFunctionSigma0Ins} can be rewritten as
\begin{equation}\label{valueFunctionSigma0Inx}
V^{\eta}_{\text{\rm A}^{\text{\rm Put}}}(x)=\sup_{u>0} \left(K -
\frac{u\varphi}{\varphi+1}\right)\left(\mathcal{Z}^{(\eta_u)}\left(x - \log u\right)- c_{\mathcal{Z}^{(\eta)}/\mathcal{W}^{(\eta)}}\mathcal{W}^{(\eta_u)}\left(x - \log u\right)\right),
\end{equation}
where $x = \log s$ and $\eta_u(x - \log u) = \eta(x) = C e^x$.

To find a closed form of value function \eqref{valueFunctionSigma0Inx} we need to identify the scale functions $\mathcal{W}^{(\eta_u)}(x - \log u)$ and $\mathcal{Z}^{(\eta_u)}(x - \log u)$.
From Theorem \ref{mainTheoremODE} it follows that both ${\mathcal{W}^{(\eta)}}(x)$ and ${\mathcal{Z}^{(\eta)}}(x)$ solve the following ordinary differential equation
\begin{equation}\label{ODE_BSwithJumps}
f''(x) = (A e^x + B)f'(x) + D e^x f(x)
\end{equation}
with $A = \frac{C}{\zeta}$, $B = \frac{\lambda-\varphi\zeta}{\zeta}$ and $D = C\frac{1 + \varphi}{\zeta}$, while the initial conditions are as follows
\begin{equation}\label{initialW}
\begin{cases}
{\mathcal{W}^{(\eta)}}(0) = \frac{1}{\zeta}, \\
{\mathcal{W}^{(\eta)}}'(0) = \frac{C+\lambda}{\zeta^2}
\end{cases}
\end{equation}
and
\begin{equation}\label{initialZ}
\begin{cases}
{\mathcal{Z}^{(\eta)}}(0) = 1, \\
{\mathcal{Z}^{(\eta)}}'(0) = \frac{C}{\zeta}.
\end{cases}
\end{equation}
Substituting $t = Ae^x$ and $F(t) = f(x)$ to \eqref{ODE_BSwithJumps} we obtain the Kummer's equation of the form
\begin{equation}\label{KummersEquation}
t F''(t) + (b - t) F'(t) - a F(t) = 0,
\end{equation}
where $b = 1 - B$ and $a = \frac{D}{A}$.

If $b$ is not an integer, then the general solution to \eqref{KummersEquation} has the form
\begin{equation}\label{KummersEquationSolution}
F(t) = K_1 {_1 F_1(a_1, b_1; t)} + K_2 t^{1-b} {_1 F_1(a_2, b_2; t)},
\end{equation}
where $K_1$ and $K_2$ are the constants that can be found based on the initial conditions, $a_1 = a$, $b_1 = b$, $a_2 = a-b+1$, $b_2 = 2-b$, while ${_1F_1}(\cdot, \cdot; \cdot)$ is the Kummer confluent hypergeometric function.

We denote by $K^W_1$, $K^W_2$ and $K^Z_1$, $K^Z_2$ the constants corresponding to ${\mathcal{W}^{(\eta)}}(x)$ and ${\mathcal{Z}^{(\eta)}}(x)$, respectively.
Using initial conditions \eqref{initialW} and \eqref{initialZ} we can simply calculate these constants for both ${\mathcal{W}^{(\eta)}}(x)$ and ${\mathcal{Z}^{(\eta)}}(x)$. By shifting these functions by $\log u$ we simply produce
$\mathcal{W}^{(\eta_u)}(x - \log u)$ and
$\mathcal{Z}^{(\eta_u)}(x - \log u)$.

The asymptotic behaviour of ${_1F_1}\left(a, b; t\right)$ for $t\rightarrow\infty$ is as follows
\begin{equation}\label{asymptoticOf1F1}
{_1 F_1(a, b; t)} = \frac{\Gamma(b)}{\Gamma(a)} e^t t^{a-b}\left[1+O\left(\frac{1}{t}\right)\right].
\end{equation}
Based on \eqref{asymptoticOf1F1} we calculate the constant $c_{\mathcal{Z}^{(\eta)}/\mathcal{W}^{(\eta)}}$ (or equivalenly $c_{\mathscr{Z}^{(\omega)}/\mathscr{W}^{(\omega)}}$) occuring in \eqref{valueFunctionSigma0Inx}. It has the following form
\begin{equation}\label{sig0Constant}
c_{\mathcal{Z}^{(\eta)}/\mathcal{W}^{(\eta)}} = \frac{K_1^Z \frac{\Gamma(b_1)}{\Gamma(a_1)}A^{a_1-b_1} + K_2^Z \frac{\Gamma(b_2)}{\Gamma(a_2)} A^{a_2-b_2}}{K_1^W \frac{\Gamma(b_1)}{\Gamma(a_1)}A^{a_1-b_1} + K_2^W \frac{\Gamma(b_2)}{\Gamma(a_2)} A^{a_2-b_2}}.
\end{equation}

Combining all the obtained results and substituting them into \eqref{valueFunctionSigma0Inx}, we can present value function \eqref{valueFunctionSigma0Inx} graphically for some sample parameter values.
This is done in Section \ref{sec:optionpricingnumerical}.
%It is exhibited in Figure \ref{...}.

\subsubsection{{\bf $\lambda = 0$}}\label{lambda0Subsection}
% 131, 566/1487, Handbook of exact solutions for ODEs

% 132, 566/1487, Handbook of exact solutions for ODEs

% 10, 585/1487!!!, Handbook of exact solutions for ODEs

% 62, 557/1487, Handbook of exact solutions for ODEs

% https://www.wolframalpha.com/input/?i=f%27%27%28x%29+%3D+B+f%27%28x%29+%2B+%28D+e%5Ex+%2B+E%29+f%28x%29

% alpha = 0:
% https://www.wolframalpha.com/input/?i=f%27%27%28x%29+%3D+-3%2F2+f%27%28x%29+%2B+%285+e%5Ex+%2B+0%29+f%28x%29

% https://www.wolframalpha.com/input/?i=f%27%27%28x%29+%3D+-2%28%5Cmu+%2B+%5Csigma%5E2%5Calpha%29%2F%28sigma%5E2%29+f%27%28x%29+%2B+%282C%2F%28sigma%5E2%29+e%5Ex+-+2%2F%28sigma%5E2%29+%28%5Cmu%5Calpha%2B%5Csigma%5E2%2F2+%5Calpha%5E2%29%29+f%28x%29

Let us consider the case of $\lambda=0$. In this case the asset price process $S_t$ enters the interval $(0,u]$ in a continuous way only.
Therefore, from Theorem \ref{mainTheorem}
\begin{equation}\label{valueFunctionlam0Ins}
\begin{aligned}
V^{\omega}_{\text{\rm A}^{\text{\rm Put}}}(s) = \sup_{u>0}\bigg\{(K-u)
\left(\lim_{\alpha\rightarrow\infty} \left(\frac{s}{u}\right)^{\alpha}\left(\mathscr{Z}^{(\omega_u^\alpha)}_{\alpha}\left(\frac{s}{u}\right) -
c_{\mathscr{Z}^{(\omega^\alpha)}_{\alpha}/\mathscr{W}^{(\omega^\alpha)}_{\alpha}}
\mathscr{W}^{(\omega_u^\alpha)}_{\alpha}\left(\frac{s}{u}\right)\right)\right)\bigg\}
\end{aligned}
\end{equation}
which is equivalent to
\begin{equation}\label{valueFunctionlam0Inx}
V^{\eta}_{\text{\rm A}^{\text{\rm Put}}}(x) = \sup_{u>0}\bigg\{(K-u)
\left(\lim_{\alpha\rightarrow\infty}e^{\alpha(x - \log u)}\left(\mathcal{Z}^{(\eta_u^\alpha)}_{\alpha}(x-\log u) -
c_{\mathcal{Z}^{(\eta^\alpha)}_{\alpha}/\mathcal{W}^{(\eta^\alpha)}_{\alpha}}
\mathcal{W}^{(\eta_u^\alpha)}_{\alpha}(x - \log u)\right)\right)\bigg\}.
\end{equation}
It suffices to find now $\mathcal{W}^{(\eta_u^\alpha)}_{\alpha}(x - \log u)$ and
$\mathcal{Z}^{(\eta_u^\alpha)}_{\alpha}(x - \log u)$.
From Theorem \ref{mainTheoremODE} it follows that $\mathcal{W}^{(\eta^\alpha)}_{\alpha}(x)$ and $\mathcal{Z}^{(\eta^\alpha)}_{\alpha}(x)$ solve
\begin{equation}\label{ODE_BSwithoutJumps}
f''(x) = B_{\alpha} f'(x) + (D_{\alpha} e^x + E_{\alpha}) f(x)
\end{equation}
with $B_{\alpha} = -\frac{2}{\sigma^2}(\zeta+\sigma^2\alpha)$, $D_{\alpha} = \frac{2C}{\sigma^2}$ and $E_{\alpha} = -\frac{2}{\sigma^2}\left(\zeta\alpha + \frac{\sigma^2}{2}\alpha^2\right)$. The initial conditions have the following form
\begin{equation}\label{initialW2}
\begin{cases}
{\mathcal{W}^{(\eta^\alpha)}_{\alpha}}(0) = 0, \\
{\mathcal{W}^{(\eta^\alpha)}_{\alpha}}'(0) = \frac{2}{\sigma^2}
\end{cases}
\end{equation}
and
\begin{equation}\label{initialZ2}
\begin{cases}
{\mathcal{Z}^{(\eta_u^\alpha)}_{\alpha}}(0) = 1, \\
{\mathcal{Z}^{(\eta_u^\alpha)}_{\alpha}}'(0) = 0.
\end{cases}
\end{equation}
Substituting $t = 2\sqrt{-D_{\alpha}e^{x}}$ and $F(t) = e^{-\frac{B_{\alpha} x}{2}} f(x)$ to \eqref{ODE_BSwithoutJumps} we obtain the Bessel differential equation of the form
\begin{equation}\label{besselEquation}
t^2 F''(t) + t F'(t) + (t^2 - v^2)F(t) = 0,
\end{equation}
where $v = \sqrt{{B_{\alpha}}^2 + 4E_{\alpha}} = \frac{2\zeta}{\sigma^2}$.
The general solution to \eqref{besselEquation} is equal to
\begin{equation*}
F(t) = K_1 J_v(t) + K_2 Y_v(t)
\end{equation*}
and therefore
\begin{equation}\label{lambda0ScaleFunction}
f(x) = e^{\frac{B_{\alpha}x}{2}}\bigg(K_1 J_v(2\sqrt{-D_{\alpha}e^x}) + K_2 Y_v(2\sqrt{-D_{\alpha}e^x})\bigg).
\end{equation}
Based on the form of \eqref{lambda0ScaleFunction} and the fact that $D_{\alpha}$ does not depend on $\alpha$ we can simply note that value function \eqref{valueFunctionlam0Inx} is also independent of $\alpha$.
Therefore, we can take an arbitrary value of $\alpha$ in \eqref{ODE_BSwithoutJumps}. Thanks to this key observation, its solution \eqref{lambda0ScaleFunction} can take a simplified form. Indeed,
for $\alpha = 0$ equation \eqref{ODE_BSwithoutJumps} is equal to
\begin{equation}\label{132EquationAlpha0}
f''(x) = B_{0} f'(x) + D_{0} e^x f(x),
\end{equation}
where $B_{0} = \frac{-2\zeta}{\sigma^2}$ and $D_{0} = \frac{2C}{\sigma^2}$.
Hence, the general solution to \eqref{132EquationAlpha0} takes the following form
\begin{equation}\label{lambda0ScaleFunctionAlpha0}
f(x) = e^{\frac{B_{0}x}{2}}\bigg(K_1 J_v(2\sqrt{-D_{0}e^x}) + K_2 Y_v(2\sqrt{-D_{0}e^x})\bigg).
\end{equation}
For $B_0 = \frac{1}{2} - n$, where $n\in\mathbb{N}_0$ and $D_0 t>0$, equation \eqref{lambda0ScaleFunctionAlpha0} reduces to
\begin{equation*}
f(x) = K_1 \bigg(\cosh(4\sqrt{D_0 e^x})\bigg)^n+ K_2 \bigg(\sinh(4\sqrt{D_0 e^x})\bigg)^n.
\end{equation*}
If we take the following sample parameters $r = 0.05$ and $\sigma = 0.2$, we obtain $n=2$ and therefore
\begin{equation}\label{WalphaZalpha}
\begin{aligned}
f(x) &= K_1 \bigg(\frac{3 \sinh(2 \sqrt{e^x})}{4 e^{\frac{5}{2}x}} + \frac{\sinh(2 \sqrt{e^x})}{e^{\frac{3}{2}x}} - \frac{3 \cosh(2 \sqrt{e^x})}{2 e^{2x}}\bigg) \\ &+ K_2 \bigg(\frac{3\cosh(2 \sqrt{e^x})}{4 e^{\frac{5}{2}x}} + \frac{\cosh(2\sqrt{e^x})}{e^{\frac{3}{2}x}} - \frac{3 \sinh(2 \sqrt{e^x})}{2 e^{2x}}\bigg).
\end{aligned}
\end{equation}
Applying initial conditions \eqref{initialW2} and \eqref{initialZ2} we can simply obtain $K_1^W$, $K_2^W$ and $K_1^Z$, $K_2^Z$.
Using equality \eqref{WalphaZalpha} which hold for both $\mathcal{W}^{(\eta^\alpha)}_{\alpha}(x)$ and $\mathcal{Z}^{(\eta^\alpha)}_{\alpha}(x)$, we can calculate that
\begin{align}\label{constantLam0}
c_{\mathcal{Z}^{(\eta^\alpha)}_{\alpha}/\mathcal{W}^{(\eta^\alpha)}_{\alpha}} = c_{\mathcal{Z}^{(\eta)}/\mathcal{W}^{(\eta)}} = \frac{K_1^Z + K_2^Z}{K_1^W + K_2^W}.
\end{align}
Taking into account all the obtained results, we can obtain value function \eqref{valueFunctionlam0Inx} for sample data.
This is done in Section \ref{sec:optionpricingnumerical} as well.

\subsection{Power discount function}\label{powerOmegaNumerical}
This time, we take into account a power function of the form
$\omega(s) = C s^n$ for $n\in(0, 1]$ and $C$ being some positive contant. This case is a generalisation of a linear discount function.

\subsubsection{{\bf $\sigma = 0$}}
% Samo podstawienie jako nalezy zastosowac to moj pomysl, a nastepnie korzystam z ponizszego zrodla
% 70, 558/1487, Handbook of exact solutions for ODEs

Similarly to the case of a linear discount function, the scale functions ${\mathcal{W}^{(\eta)}}(x)$ and ${\mathcal{Z}^{(\eta)}}(x)$ solve
\begin{equation}\label{ODE_BSwithJumps2}
f''(x) = (A e^{nx} + B)f'(x) + D e^{nx} f(x)
\end{equation}
with $A = \frac{C}{\zeta}$, $B = \frac{\lambda-\varphi\zeta}{\zeta}$ and $D = C\frac{n + \varphi}{\zeta}$, while the initial conditions are the same as those provided in \eqref{initialW} and \eqref{initialZ}.
Applying a substitution $t = \frac{A}{n} e^{nx}$ and $F(t) = f(x)$ we transform \eqref{ODE_BSwithJumps2} into
\begin{equation}\label{KummersEquation2}
t F''(t) + (b - t) F'(t) - a F(t) = 0,
\end{equation}
where $b = 1 - \frac{B}{n}$ and $a = \frac{D}{An}$.
The general solution to \eqref{KummersEquation2} has the same form as was provided in \eqref{KummersEquationSolution}.

Therefore, for both the linear and the power discount function $\omega$, the form of the value function $V^{\eta}_{\text{\rm A}^{\text{\rm Put}}}(x)$ is identical.

\subsubsection{{\bf $\lambda = 0$}}
As in the above case, the idea of finding a closed form of the value function can be borrowed from the linear case.
This time, the scale functions $\mathcal{W}^{(\eta^\alpha)}_{\alpha}(x)$ and $\mathcal{Z}^{(\eta^\alpha)}_{\alpha}(x)$ satisfy the equation
\begin{equation*}
f''(x) = B_{\alpha}f'(x) + (D_{\alpha} e^{nx} + E_{\alpha}) f(x)
\end{equation*}
with $B_{\alpha} = -\frac{2}{\sigma^2}(\zeta+\sigma^2\alpha)$, $D_{\alpha} = \frac{2C}{\sigma^2}$ and $E_{\alpha} = -\frac{2}{\sigma^2}\left(\zeta\alpha + \frac{\sigma^2}{2}\alpha^2\right)$, while the initial conditions are of the form \eqref{initialW2} and \eqref{initialZ2}.
If we substitute $t = \frac{2}{n}\sqrt{-D_{\alpha} e^{nx}}$ and $F(t) = e^{-\frac{B_{\alpha}x}{2}} f(x)$ we receive the Bessel differential equation for $F(t)$ with the solution
\begin{equation*}
F(t) = K_1 J_v(t) + K_2 Y_v(t).
\end{equation*}
Therefore, we have
\begin{equation}\label{lambda0PowerScaleFunction}
f(x) = e^{\frac{B_{\alpha}x}{2}}\bigg(K_1 J_v\bigg(\frac{2}{n}\sqrt{-D_{\alpha}e^{nx}}\bigg) + K_2 Y_v\bigg(\frac{2}{n}\sqrt{-D_{\alpha}e^{nx}}\bigg)\bigg),
\end{equation}
where $v = \frac{\sqrt{{B_{\alpha}}^2 + 4 E_{\alpha}}}{n} = \frac{2\zeta}{n \sigma^2}$.
We can show, as in the prevoius section, that the value function which arises in this scenario does not depend on $\alpha$. Thus, for $\alpha = 0$, \eqref{lambda0PowerScaleFunction} takes the form
\begin{equation*}
f(x) = e^{\frac{B_{0}x}{2}}\bigg(K_1 J_v\bigg(\frac{2}{n}\sqrt{-D_{0}e^{nx}}\bigg) + K_2 Y_v\bigg(\frac{2}{n}\sqrt{-D_{0}e^{nx}}\bigg)\bigg).
\end{equation*}
Again, having exact formulas for the scale functions, we can easily represent the form of the value function.
All results are presented in Section \ref{sec:optionpricingnumerical}.

\section{Option pricing -- numerical approach}\label{sec:optionpricingnumerical}

In this section we show how to numerically identify the value function $V^{\omega}_{\text{\rm A}^{\text{\rm Put}}}(s)$ for arbitrary discount function $\omega$.
We present some figures corresponding to the various discount functions as well.

\subsection{Different discount functions}
For some discount functions $\omega$ we are unable to find the analytical forms of $\mathcal{W}^{(\eta)}(x)$, $\mathcal{Z}^{(\eta)}(x)$, $\mathcal{W}^{(\eta^\alpha)}_{\alpha}(x)$ and $\mathcal{Z}^{(\eta^\alpha)}_{\alpha}(x)$ which are the solutions to the ordinary differential equations occurring in Theorem \ref{mainTheoremODE}. That is, formally we cannot identify explicitly the value function either.
In such a situation, we can proceed a numerical analysis of these equations.

In general, solving a high-order ordinary differential equation consists in transforming it into first-order vector form and then applying an appropriate algorithm that returns us the numerical solution of the $n+1-$dimensional system of first-order ordinary differential equations.
%Whenever explicit solutions to ordinary differential equations do not exist or they are too hard to find, we can try to obtain the numerical ones.
For practical purposes, however -- such as in financial engineering -- numeric approximations to the solutions of ordinary differential equations are often sufficient.
In this paper we focus on the Higher-Order Taylor Method.
This method employs the Taylor polynomial of the solution to the equation. It approximates the $0-$th order term by using the previous step's value (which is the initial condition for the first step), and the subsequent terms of the Taylor expansion by using the differential equation.

%We also present some examples of this approach and show that numerical solutions are compatible with their analytical equivalents.

\subsubsection{{\bf $\sigma = 0$}}
For the case of $\sigma = 0$ we show in the previous section how to derive the value function for the linear discount function $\omega$.
Figure \ref{sig0_valueFunction} presents comparison of the value function given in \eqref{valueFunctionSigma0Ins} when the scale functions were calculated analytically (as shown in subsection \ref{sigma0Subsection}) and numerically by solving differential ordinary equation \eqref{ODE_BSwithJumps}. In this example, a linear discount function was chosen, i.e. $\omega(s) = Cs$.
The difference between these functions is so small that this is negligible.

\begin{figure}[ht]
\centering
\includegraphics[width=13cm]{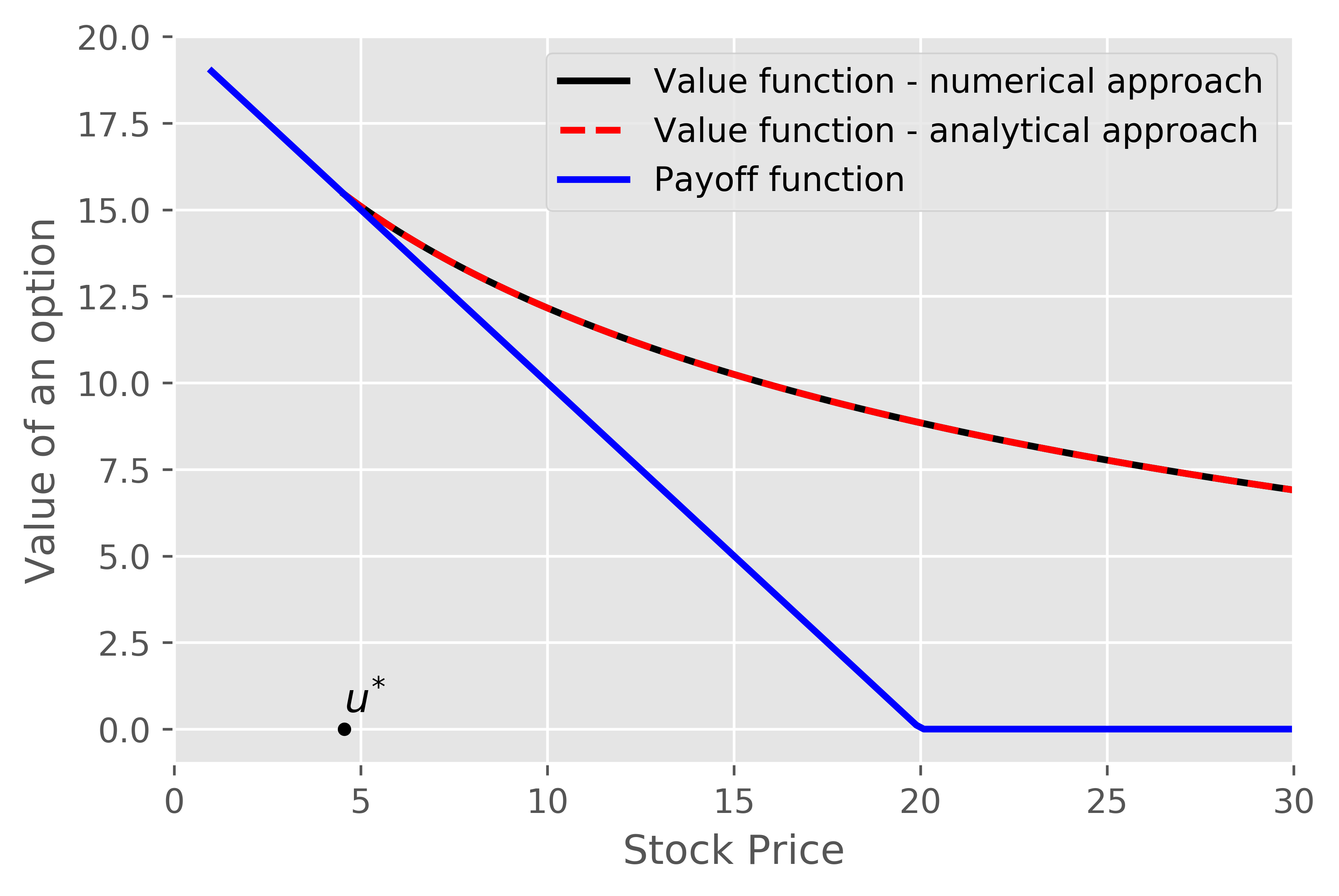}
\caption{Comparison of value function \eqref{valueFunctionSigma0Ins} for $\omega(s)=Cs$ for both methods of determining the scale functions -- analytical and numerical one. The chosen set of parameters is as follows: $K = 20$, $C = 0.1$, $r = 0.05$, $\lambda = 6$, $\varphi = 2$.}
\label{sig0_valueFunction}
\end{figure}

Moreover, Figure \ref{sig0_c} illustrates the constant $c_{\mathcal{Z}^{(\eta)}/\mathcal{W}^{(\eta)}}$ obtained in \eqref{sig0Constant} together with the quotient of the functions $\mathcal{Z}^{(\eta)}(x)$ and $\mathcal{W}^{(\eta)}(x)$.

\begin{figure}[ht]
\centering
\includegraphics[width=13cm]{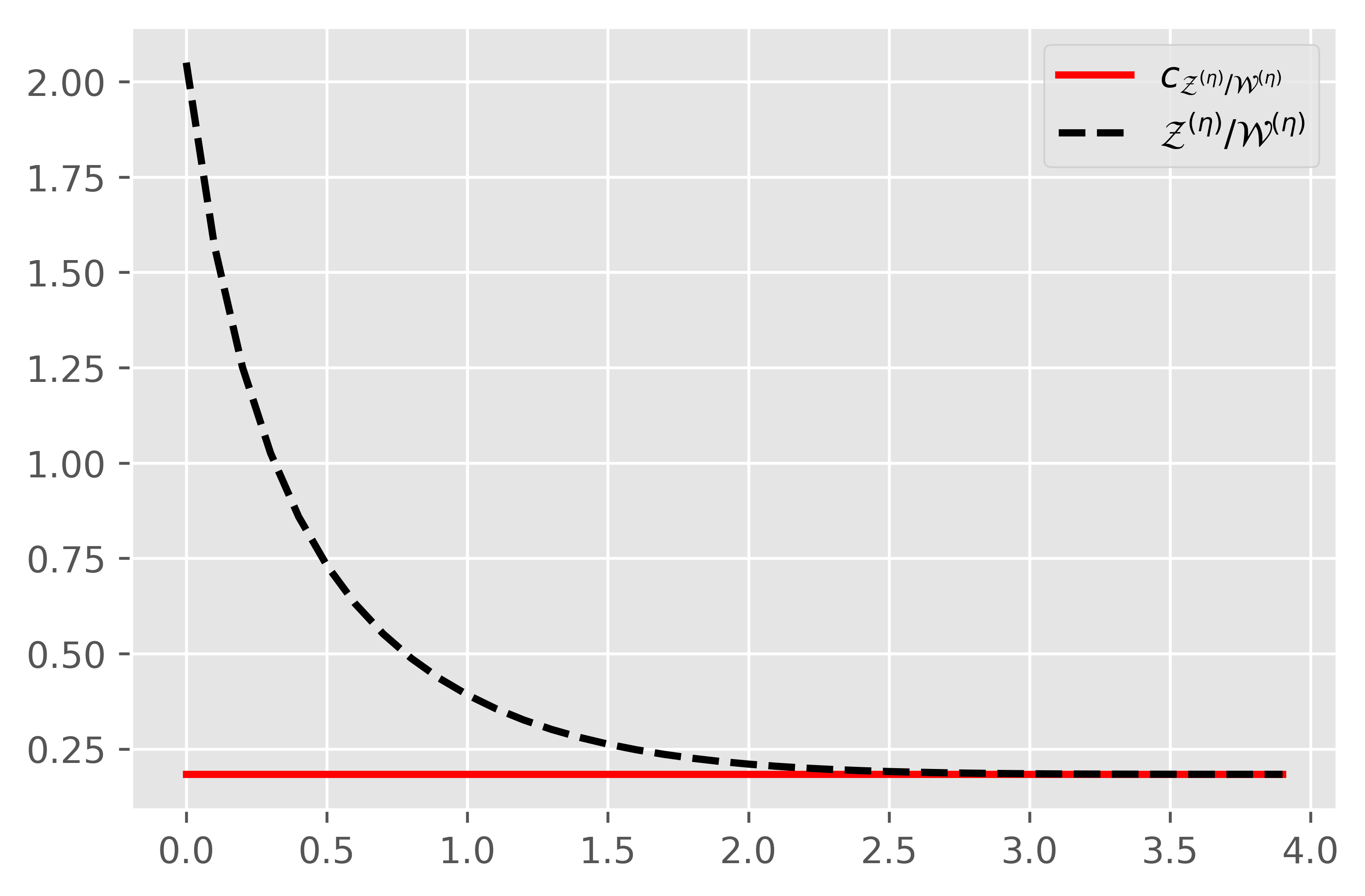}
\caption{Comparison of the constant $c_{\mathcal{Z}^{(\eta)}/\mathcal{W}^{(\eta)}}$ and the ratio of $\mathcal{Z}^{(\eta)}(x)$ and $\mathcal{W}^{(\eta)}(x)$ for a linear discount function $\omega(s)=Cs$ and $K = 20$, $C = 0.1$, $r = 0.05$, $\lambda = 6$, $\varphi = 2$.}
\label{sig0_c}
\end{figure}

As we mentioned at the beginning of this section, we are not always able to get an analytical solution to an ordinary differential equation.
This is the case for example when the discount function is of the form $\omega(s) = C \arctan(s)$ for some positive $C$.
Then we can only obtain the scale functions numerically.
Figure \ref{sig0_s_And_arctans} shows two value functions, for $\omega(s) = Cs$ and $\omega(s) = C \arctan(s)$, respectively.
Since for all positive $s$ we have $s>\arctan(s)$, then we expect that the value function corresponding to $\omega(s) = C \arctan(s)$ takes greater values rather than for $\omega(s) = Cs$. We can also note that the difference between these functions becomes greater for higher values of $s$, which is in line with economical intuition since the difference between $\omega(s) = s$ and $\omega(s) = \arctan (s)$ also extends as $s$ increases.

\begin{figure}[ht]
\centering
\includegraphics[width=13cm]{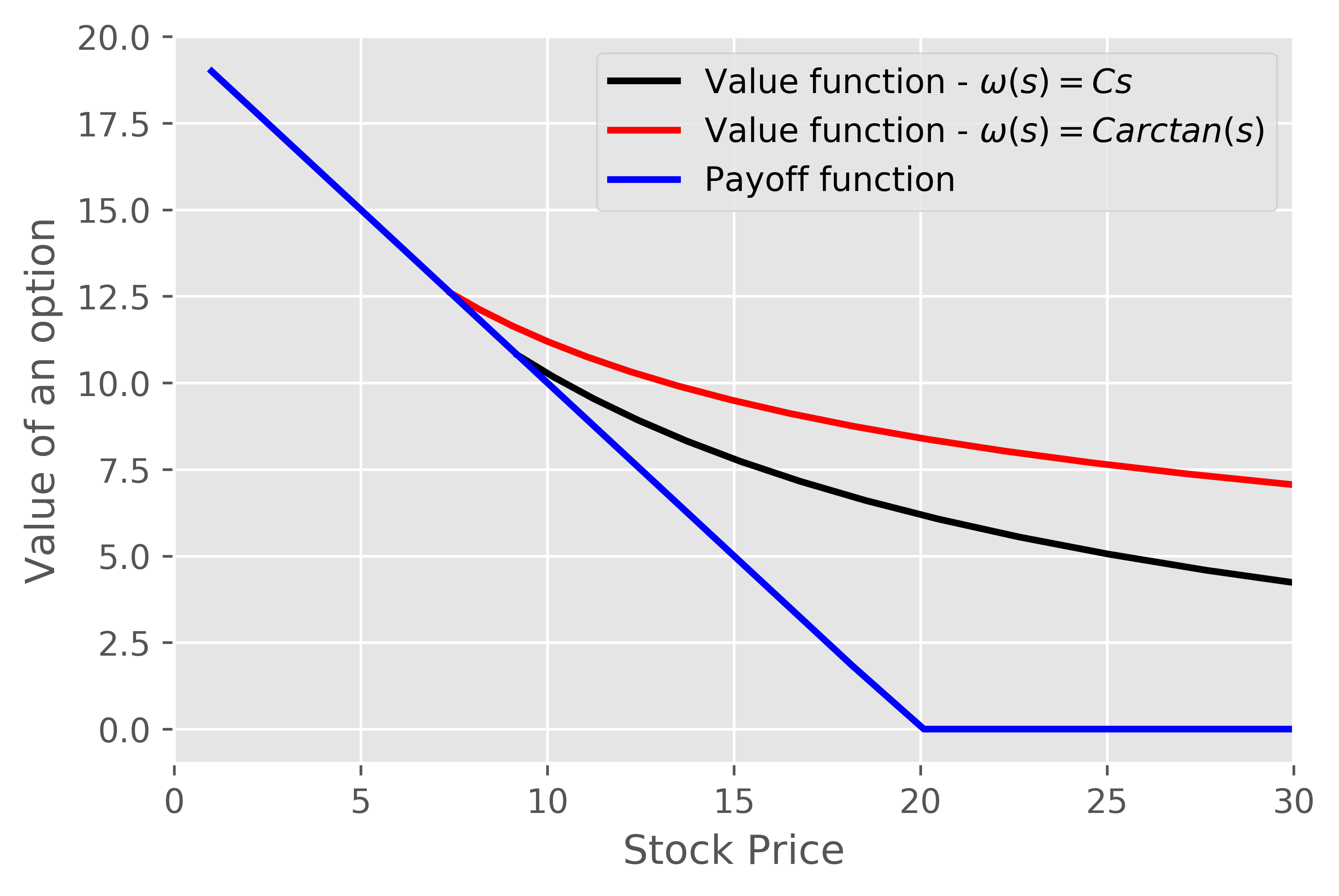}
\caption{Comparison of value function \eqref{valueFunctionSigma0Ins} for both $\omega(s)=Cs$ and $\omega = C\arctan(s)$ and for the given set of parameters: $K = 20$, $C = 0.5$, $r = 0.05$, $\lambda = 6$, $\varphi = 2$.}
\label{sig0_s_And_arctans}
\end{figure}

\subsubsection{{\bf $\lambda = 0$}}
The case of $\lambda = 0$ corresponds to the situation when stock price process \eqref{S_t} does not have any jumps.
Then the value function takes the form \eqref{lambda=0}.
From the numerical point of view, the problem lies in choosing a sufficiently large value of $\alpha$ in \eqref{lambda=0} to obtain the final and right form of the value function.
In this section we avoid this problem by selecting discount functions for which the value function is independent of the $\alpha$ parameter.

Figure \ref{lam0_valueFunction} presents comparison of the value function given in \eqref{valueFunctionlam0Ins} for
$\omega(s) = Cs$ and for the scale functions obtained analytically and numerically.
As shown in subsection \ref{lambda0Subsection}, in this case we can take an arbitrary value of $\alpha$ and obtain a simplified form of the value function and ordinary differential equation that the scale functions solve.
Similarly to the previous example, we again can observe a negligible difference between these two value functions.

\begin{figure}[ht]
\centering
\includegraphics[width=13cm]{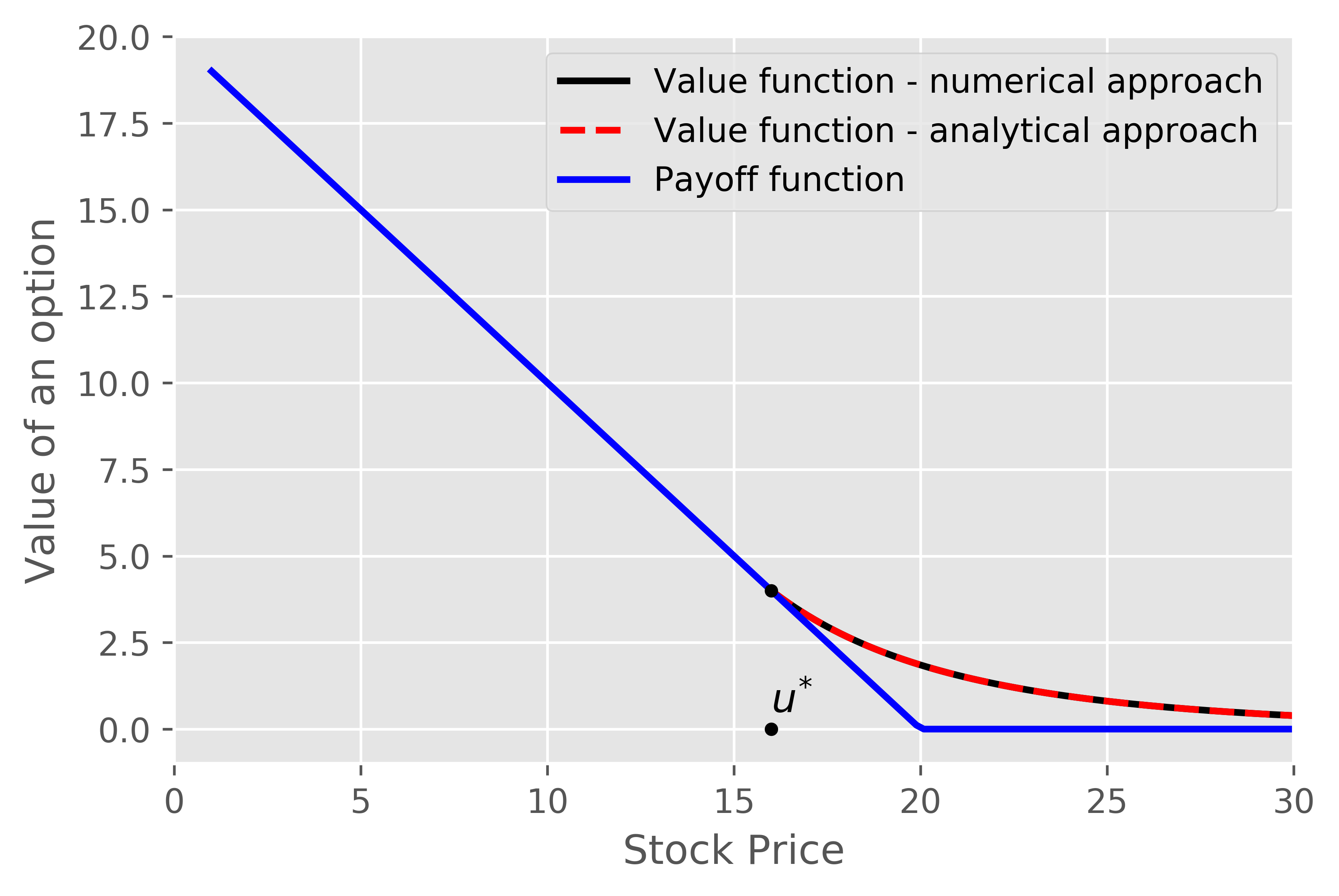}
\caption{Comparison of value function \eqref{valueFunctionlam0Ins} for $\omega(s)=Cs$ for both methods of determining the scale functions -- analytical and numerical one. The chosen set of parameters is as follows: $K = 20$, $C = 0.1$, $r = 0.05$, $\sigma = 0.2$.}
\label{lam0_valueFunction}
\end{figure}

In turn, Figure \ref{lam0_c} shows the constant $c_{\mathcal{Z}^{(\eta)}/\mathcal{W}^{(\eta)}}$ given in \eqref{constantLam0} with the ratio of $\mathcal{Z}^{(\eta)}(x)$ and $\mathcal{W}^{(\eta)}(x)$.

\begin{figure}[ht]
\centering
\includegraphics[width=13cm]{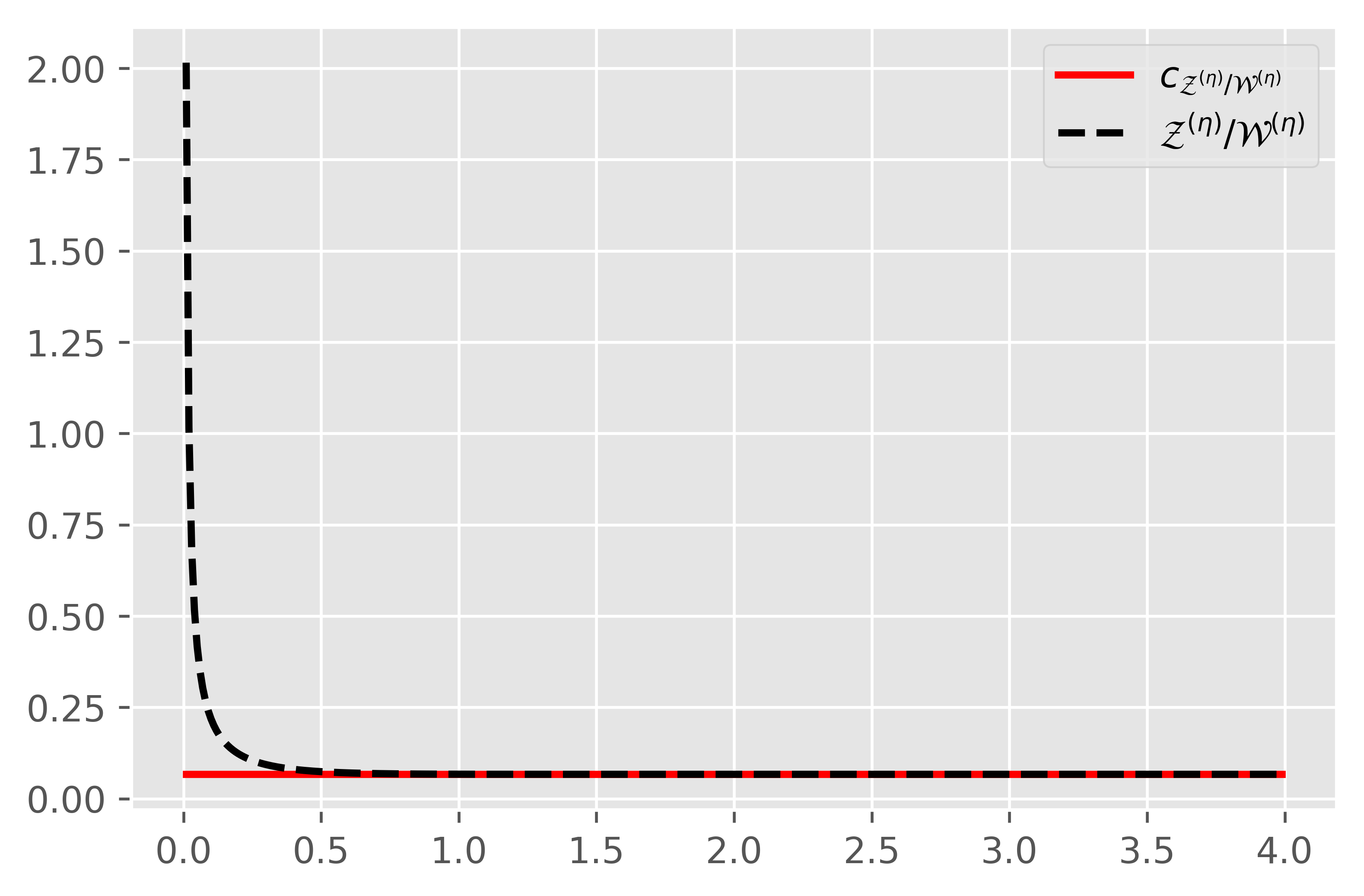}
\caption{Comparison of the constant $c_{\mathcal{Z}^{(\eta)}/\mathcal{W}^{(\eta)}}$ and the ratio of $\mathcal{Z}^{(\eta)}(x)$ and $\mathcal{W}^{(\eta)}(x)$ for a linear discount function $\omega(s)=Cs$ and $K = 20$, $C = 0.1$, $r = 0.05$, $\sigma = 0.2$.}
\label{lam0_c}
\end{figure}

Lastly, in Figure \ref{lam0_sqrts_and_asqrtsPlusZ} we can observe the value functions for both $\omega(s) = C\sqrt{s}$ and $\omega(s) = C\sqrt{s} + Z$ for some positive $Z$, i.e. we compare two discount functions that differ in shift.
This time, we can see that the value functions obtained in this way also differ only in shift, which confirms financial intuition.

\begin{figure}[ht]
\centering
\includegraphics[width=13cm]{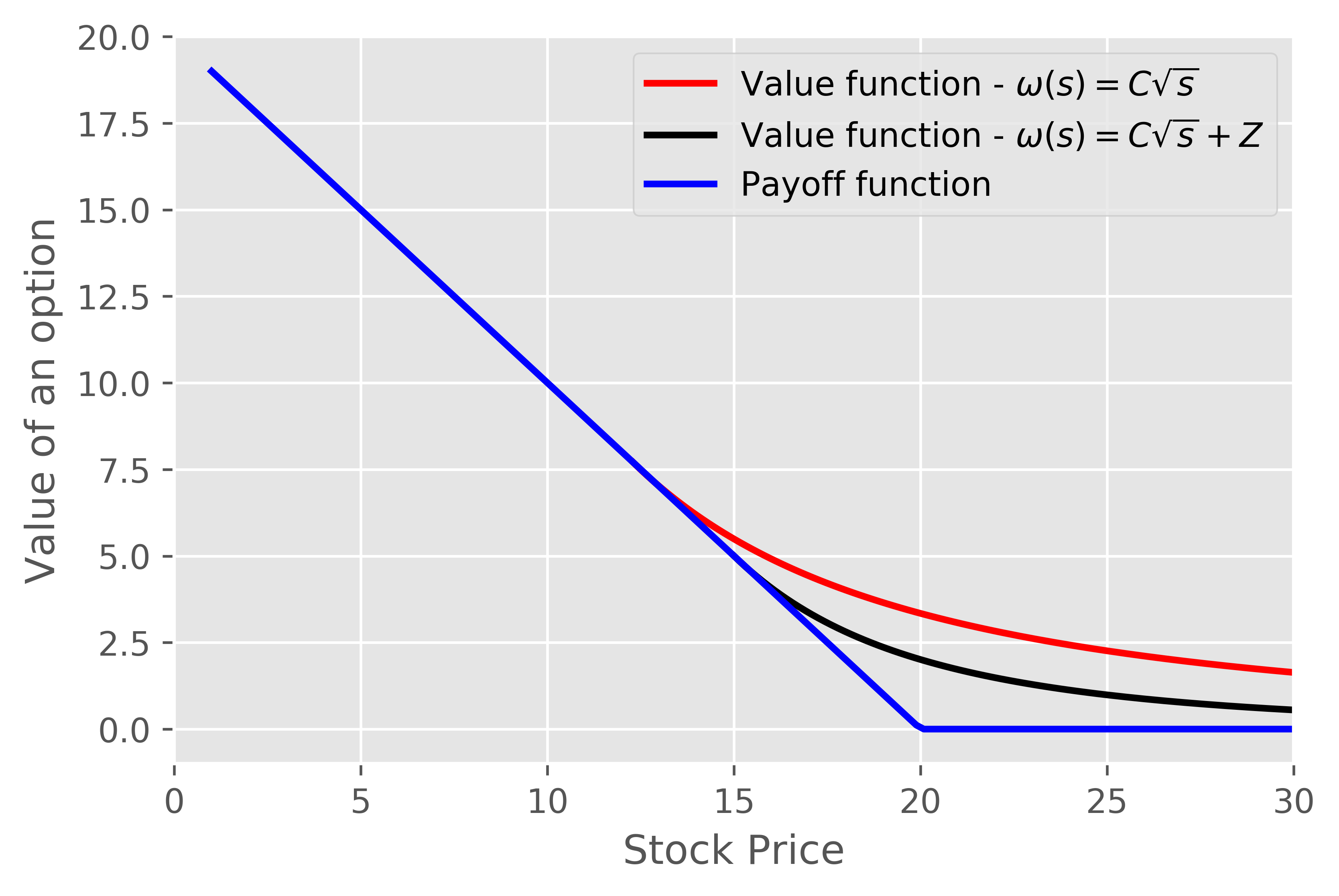}
\caption{Comparison of value function \eqref{valueFunctionlam0Ins} for both $\omega(s)=C\sqrt{s}$ and $\omega = C\sqrt{s}+Z$ and for the given set of parameters: $K = 20$, $C = 0.005$, $Z = 0.1$, $r = 0.05$, $\sigma = 0.2$.}
\label{lam0_sqrts_and_asqrtsPlusZ}
\end{figure}

\subsubsection{{\bf $\sigma > 0$ and $\lambda >0$}}
% 139, 251/802, Handbook of exact solutions for ODEs
% 108, 247/802, Handbook of exact solutions for ODEs
The most general case is when $\sigma>0$ and $\lambda>0$.
Then the considered value function $V^{\omega}_{\text{\rm A}^{\text{\rm Put}}}(s)$ is given by \eqref{sigma>0}. It can be also represented as the function of $x$ variable:
\begin{equation}\label{valueFunctionSigma>0Inx}
\begin{aligned}
&V^{\eta}_{\text{\rm A}^{\text{\rm Put}}}(x) = \sup_{u>0}\bigg\{\left(K - \frac{u \varphi}{\varphi+1}\right)\left(\mathcal{Z}^{(\eta_u)}(x - \log u)- c_{\mathcal{Z}^{(\eta)}/\mathcal{W}^{(\eta)}} \mathcal{W}^{(\eta_u)}(x - \log u)\right)
\\ &+ (K-u)
\left(\lim_{\alpha\rightarrow\infty} e^{\alpha(x-\log u)}\left(\mathcal{Z}^{(\eta_u^\alpha)}_{\alpha}(x - \log u) -
c_{\mathcal{Z}^{(\eta^\alpha)}_{\alpha}/\mathcal{W}^{(\eta^\alpha)}_{\alpha}}
\mathcal{W}^{(\eta_u^\alpha)}_{\alpha}(x - \log u)\right)\right)\bigg\}.
\end{aligned}
\end{equation}
For the linear discount function $\omega(s)=Cs$, the scale functions $\mathcal{W}^{(\eta)}(x)$ and $\mathcal{Z}^{(\eta)}(x)$ occurring in \eqref{valueFunctionSigma>0Inx} are the solutions to the following ordinary differential equation
\begin{equation}\label{thirdOrderODEf(x)}
\begin{split}
f'''(x) &= Af''(x) + (Be^x + D)f'(x) + E e^x f(x)
\end{split}
\end{equation}
with $A = \gamma_2 + \gamma_3$, $B = C\left[\Upsilon_2(\gamma_2-\gamma_3)-\Upsilon_1\gamma_3\right]$, $D = -\gamma_2\gamma_3$, $E = C\left[\Upsilon_2(\gamma_2-\gamma_3)+\Upsilon_1\gamma_2\gamma_3-\Upsilon_1\gamma_3\right]$.

The initial conditions for $\mathcal{W}^{(\eta)}(x)$ and $\mathcal{Z}^{(\eta)}(x)$ are as follows
\begin{equation*}
\begin{cases}
{\mathcal{W}^{(\eta)}}(0) = 0, \\
{\mathcal{W}^{(\eta)}}'(0) = \Upsilon_2\gamma_2 + \Upsilon_3\gamma_3, \\
{\mathcal{W}^{(\eta)}}''(0) = \Upsilon_2 {\gamma_2}^2 + \Upsilon_3{\gamma_3}^2
\end{cases}
\end{equation*}
and
\begin{equation*}
\begin{cases}
{\mathcal{Z}^{(\eta)}}(0) = 1, \\
{\mathcal{Z}^{(\eta)}}'(0) = 0, \\
{\mathcal{Z}^{(\eta)}}''(0) = C\left[\Upsilon_2(\gamma_2-\gamma_3) - \Upsilon_1\gamma_3\right].
\end{cases}
\end{equation*}
Thus $\mathcal{W}^{(\eta^\alpha)}_{\alpha}(x)$ and $\mathcal{Z}^{(\eta^\alpha)}_{\alpha}(x)$ solve
\begin{equation}\label{thirdOrderODEf(x)2}
\begin{split}
f'''(x) &= A_{\alpha}f''(x) + (B_{\alpha}e^x + D_{\alpha})f'(x) + (E_{\alpha} e^x + F_{\alpha})f(x)
\end{split}
\end{equation}
with $A_{\alpha} = \gamma_{{\alpha}_2} + \gamma_{{\alpha}_3}$, $B_{\alpha} = C\left[\Upsilon_{{\alpha}_2}(\gamma_{{\alpha}_2}-\gamma_{{\alpha}_3})-\Upsilon_{{\alpha}_1}\gamma_{{\alpha}_3}\right]$, $D_{\alpha} = -\Upsilon_{{\alpha}_2}(\gamma_{{\alpha}_2}-\gamma_{{\alpha}_3})\psi(\alpha)-\gamma_2\gamma_3+\Upsilon_{{\alpha}_1}\gamma_{{\alpha}_3}\psi(\alpha)$, $E_{\alpha} = C\left[\Upsilon_{\alpha_2}(\gamma_{\alpha_2}-\gamma_{\alpha_3})+\Upsilon_{\alpha_1}\gamma_{\alpha_2}\gamma_{\alpha_3}-\Upsilon_{\alpha_1}\gamma_{\alpha_3}\right]$, $F_{\alpha} = -\Upsilon_{\alpha_1}\gamma_{\alpha_2}\gamma_{\alpha_3}\psi(\alpha)$.
The initial conditions for ${\mathcal{W}^{(\eta^\alpha)}_{\alpha}}(x)$ and ${\mathcal{Z}^{(\eta^\alpha)}_{\alpha}}(x)$ are as follows
\begin{equation*}
\begin{cases}
{\mathcal{W}^{(\eta^\alpha)}_{\alpha}}(0) = 0, \\
{\mathcal{W}^{(\eta^\alpha)}_{\alpha}}'(0) = \Upsilon_{\alpha_2}\gamma_{\alpha_2} + \Upsilon_{\alpha_3}\gamma_{\alpha_3}, \\
{\mathcal{W}^{(\eta^\alpha)}_{\alpha}}''(0) = \Upsilon_{\alpha_2} {\gamma_{\alpha_2}}^2 + \Upsilon_{\alpha_3}{\gamma_{\alpha_3}}^2
\end{cases}
\end{equation*}
and
\begin{equation*}
\begin{cases}
{\mathcal{Z}^{(\eta^\alpha)}_{\alpha}}(0) = 1, \\
{\mathcal{Z}^{(\eta^\alpha)}_{\alpha}}'(0) = 0, \\
{\mathcal{Z}^{(\eta^\alpha)}_{\alpha}}''(0) = C\left[\Upsilon_{\alpha_2}(\gamma_{\alpha_2}-\gamma_{\alpha_3}) - \Upsilon_{\alpha_1}\gamma_{\alpha_3}\right].
\end{cases}
\end{equation*}
In this case, we are not able to identify explicit solutions to third order ordinary differential equations \eqref{thirdOrderODEf(x)} and \eqref{thirdOrderODEf(x)2}.
So we are forced to use only the numerical method to find the scale functions and hence the value function.

Figure \ref{sigLargerThan0_severalValueFunctions} shows several graphs of the value function
\eqref{sigma>0} for different values of $\alpha$ parameter together with the first and second component occurring in \eqref{sigma>0}.

\begin{figure*}
	\centering
    \begin{subfigure}[b]{0.475\textwidth}
    	\centering
        \includegraphics[width=\textwidth]{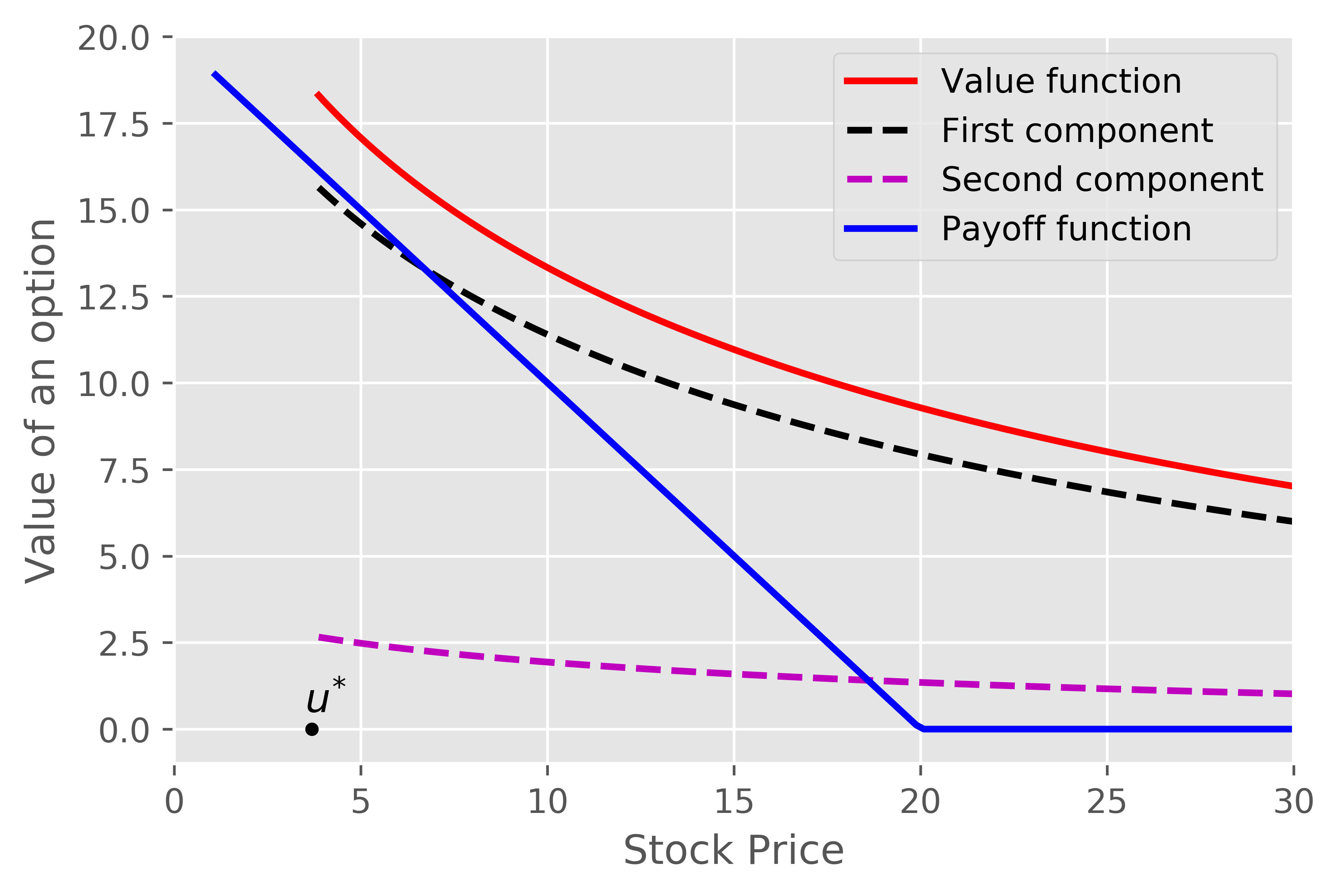}
        \caption{$\alpha = 10$}
        \label{alpha0}
    \end{subfigure}
    \hfill
    \begin{subfigure}[b]{0.475\textwidth}
        \centering
        \includegraphics[width=\textwidth]{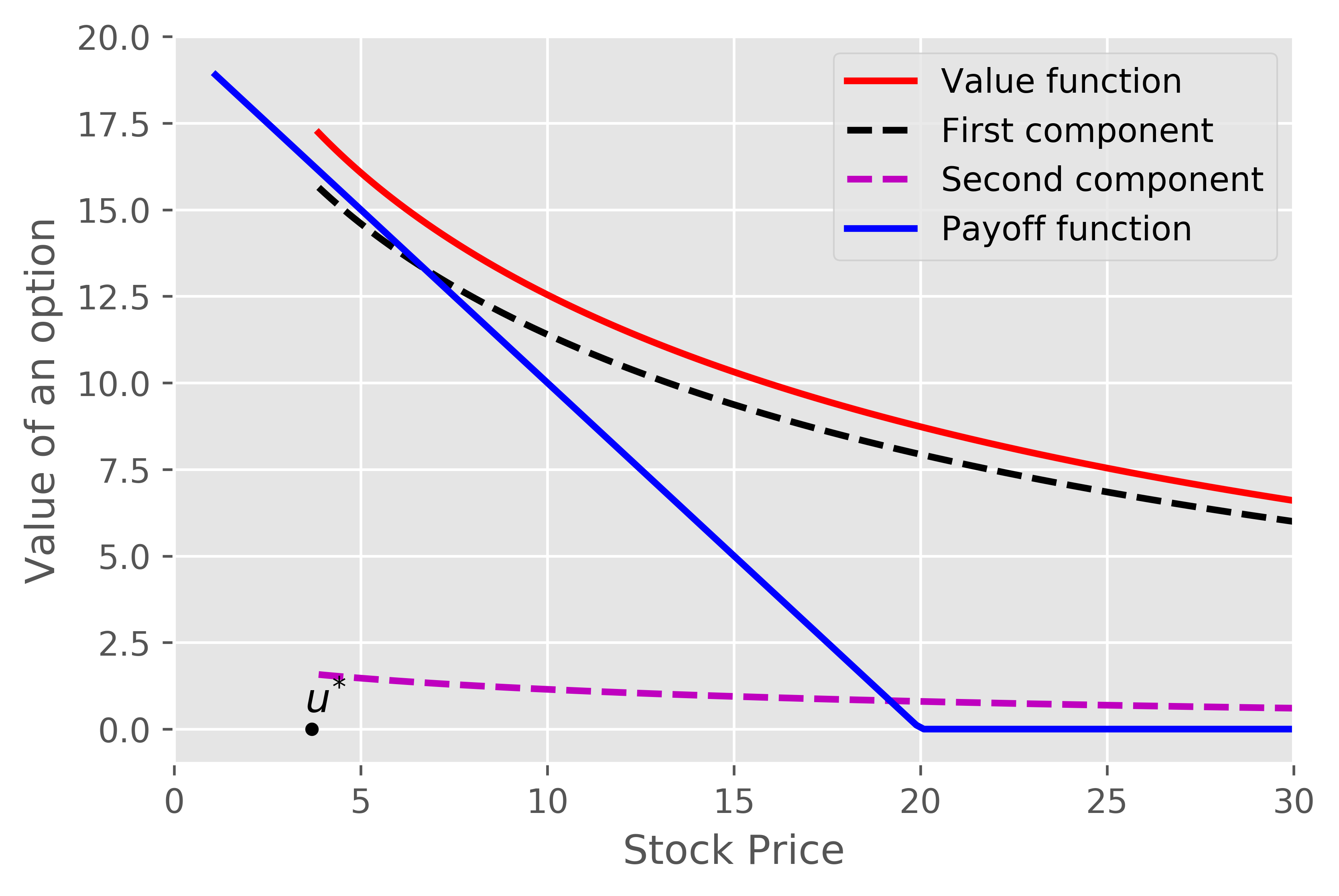}
        \caption{$\alpha = 20$}
        \label{alpha10}
    \end{subfigure}
	\vskip\baselineskip
    \begin{subfigure}[b]{0.475\textwidth}
        \centering
        \includegraphics[width=\textwidth]{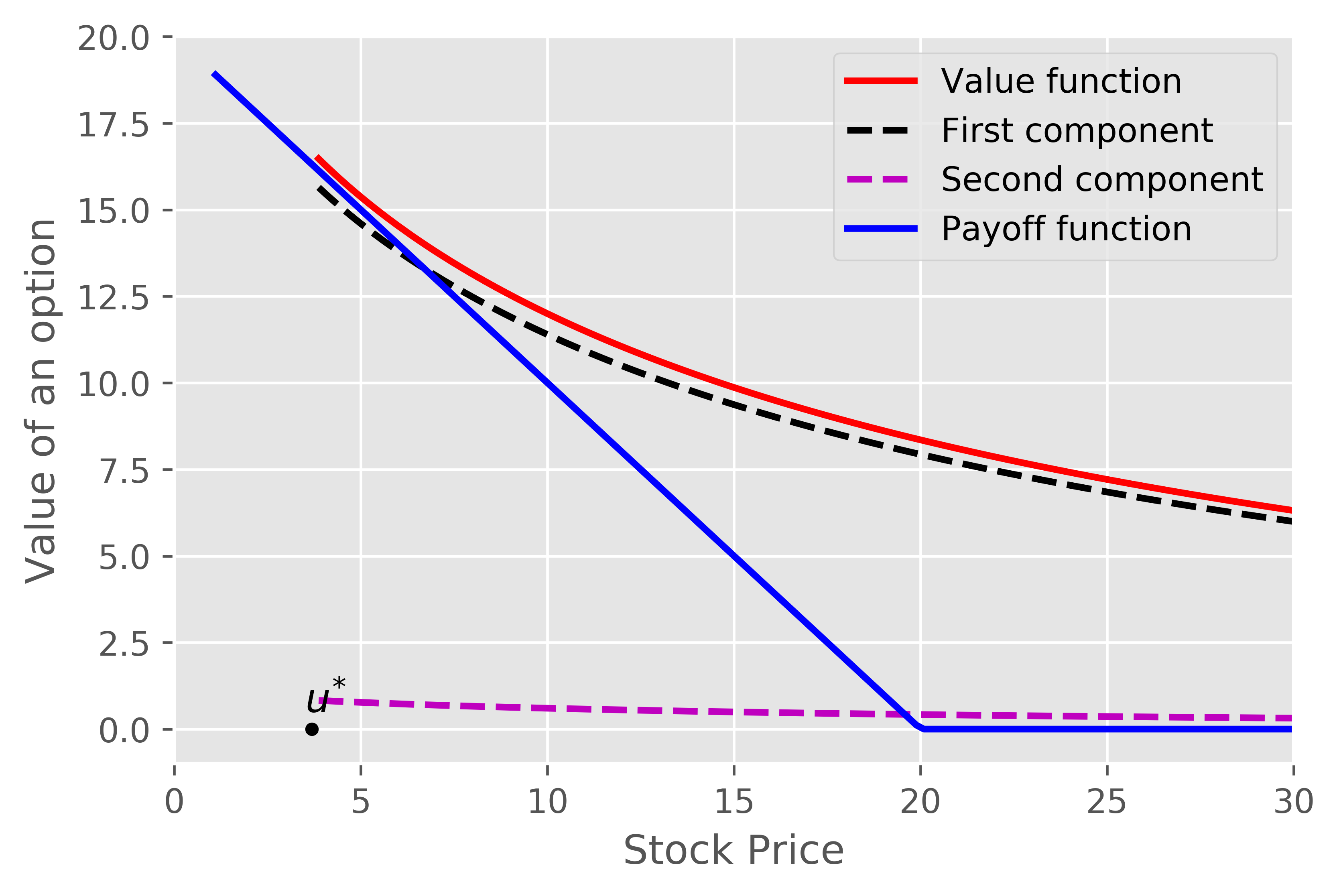}
        \caption{$\alpha = 50$}
        \label{alpha20}
    \end{subfigure}
	\hfill
    \begin{subfigure}[b]{0.475\textwidth}
        \centering
        \includegraphics[width=\textwidth]{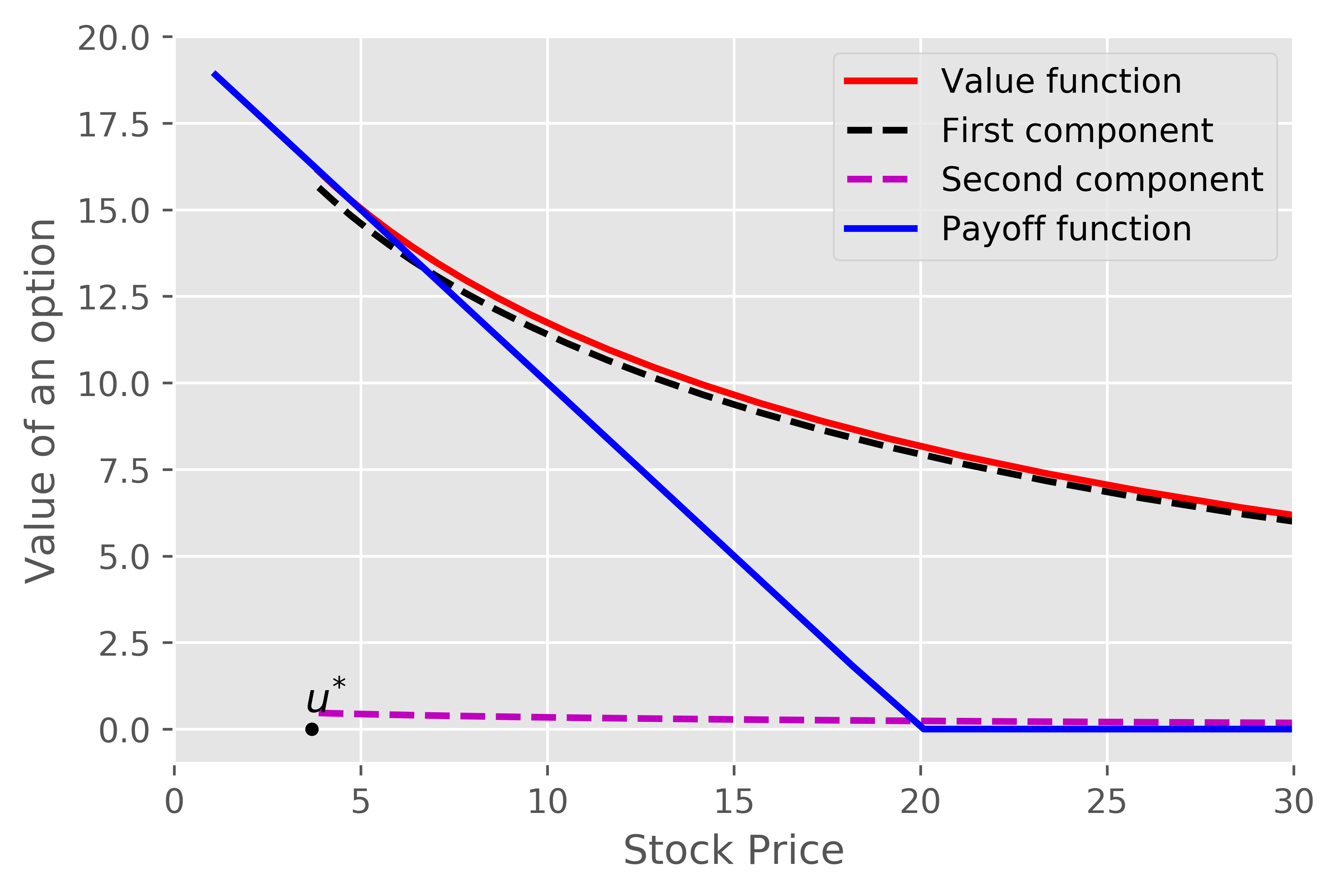}
        \caption{$\alpha = 150$}
        \label{alpha30}
    \end{subfigure}
    \caption{Comparison of value function \eqref{sigma>0} for the particular choice of $\alpha $, $\omega(s) = Cs$ and for the given set of parameters: $K = 20$, $C = 0.1$, $r = 0.05$, $\sigma = 0.2$, $\lambda = 6$, $\varphi = 2$.}
    \label{sigLargerThan0_severalValueFunctions}
\end{figure*}

\section{Conclusions}
In this paper, we have presented the novel approach to pricing the perpetual American put options with asset-dependent discounting.
For the asset price process $S_t$ being the geometric spectrally negative L\'evy process we have shown that value function \eqref{mainProblem} can be represented in a closed form based the omega-type scale functions that solve some ordinary differential equations given in Theorem \ref{mainTheoremODE}.
%These theoretical considerations allow us to find a price option for a specifically selected discount function $\omega$.
We have used these theoretical results to perform extended numerical analysis for some key financial examples.
In particular, for some cases we have managed to produce some explicit formulas for the value function. In the cases where it was impossible
to do so, we have used the numerical analysis of the above mentioned ordinary differential equations based on Higher-Order Taylor Method.
We have presented many figures of the value functions that arise in various scenarios.
%We have also presented on a few examples that the analytical and numerical methods are compatible with each other.

One can think of further generalisations. For example when the discount factor is randomised. It can be done in different ways, e.g. by introducing additional Markov economical environment.
This type of research is left for future investigations.

\end{document}